\documentclass[11pt]{article}

\usepackage{amssymb,latexsym}
\usepackage{amsmath}
\usepackage{amsthm}
\usepackage{graphics}
\usepackage{graphicx}
\usepackage{color}
\usepackage{hyperref}
\hypersetup{colorlinks=true,linkcolor=blue,citecolor=blue}
\usepackage{fullpage}
\usepackage{comment}
\usepackage{gensymb} 

\usepackage{lipsum}

\usepackage[ruled, linesnumbered, noend]{algorithm2e}
\usepackage[utf8]{inputenc}

\usepackage{etoolbox}
\usepackage[]{nomencl}   
\makenomenclature
\usepackage{parskip}

\providetoggle{nomsort}

\settoggle{nomsort}{true}

\usepackage{hyperref}

\usepackage{gensymb} 

\renewcommand{\vec}[1]{\mathbf{#1}}

\newtheorem{theorem}{Theorem}

\newtheorem{corollary}[theorem]{Corollary}

\newtheorem{definition}[theorem]{Definition}

\newtheorem{lemma}[theorem]{Lemma}
\newtheorem*{lemma*}{Lemma}

\theoremstyle{remark}
\newtheorem{remark}[theorem]{Remark}
\newtheorem*{remark*}{Remark}

\newtheorem*{open*}{Open problem}
\newtheorem{example}{Example}
\newtheorem*{example*}{Example}

\usepackage{xcolor}

\definecolor{ForestGreen}{rgb}{.05,.50,.05}
\definecolor{BrickRed}{rgb}{.80,.26,.33}
\definecolor{cadetgrey}{rgb}{0.45,0.45,0.45}

\newcommand{\C}{\mathcal C}

\newcommand{\R}{\mathbb{R}}

\newcommand{\MWW}{\mathrm{MWW}}

\newcommand{\U}{\mathcal{U}}

\newcommand{\G}{\mathcal{G}}

\renewcommand{\P}{\mathcal{P}}

\usepackage[multiple]{footmisc}
\usepackage{algorithmic}
\usepackage{tikz}
\usepackage{ifthen}

\usepackage[multiple]{footmisc}

\DeclareMathOperator*{\argmax}{arg\,max}
\title{Algorithms for Competitive Division of Chores\footnote{The authors are grateful to Nikita Kalinin, Alisa Maricheva, Herve Moulin, Ekaterina Rzhevskaya, Erel Segal-Halevi, Nisarg Shah, Inbal Talgam-Cohen for many useful discussions, and to Lillian Bluestein for the help with proofreading.
	Fedor appreciated the hospitality of the Federmann Center for the Study of Rationality, where this project was started. 
	Fedor's work was supported by the
Linde Institute at Caltech, the National Science Foundation (grant CNS 1518941), the Lady Davis Foundation, and the European Research Council (ERC) under the European Union's Horizon 2020 research and innovation program (grant agreement n$\degree$740435). A part of this work was done while Simina was visiting the Simons Institute for the Theory of Computing.}}

\author{Simina Br\^anzei\footnote{Purdue University, \href{mailto:simina.branzei@gmail.com}{simina.branzei@gmail.com}, \href{http://simina.info}{http://simina.info}} \and Fedor Sandomirskiy\footnote{Princeton University, USA, \href{mailto:fsandomi@princeton.edu}{fsandomi@princeton.edu}, \href{https://www.fedors.info}{{https://www.fedors.info}}}
}

\date{}

\begin{document}

\maketitle
\begin{abstract}
We study the problem of allocating divisible bads (chores) among multiple agents with additive utilities when monetary transfers are not allowed. 
		The competitive rule is known for its remarkable fairness and efficiency properties in the case of goods.
			This rule was extended to chores in prior work by Bogomolnaia, Moulin, Sandomirskiy, and Yanovskaya~\cite{BMSY17}. The rule produces Pareto optimal and envy-free allocations for both goods and chores.
		In the case of goods, the outcome of the competitive rule can be easily computed. Competitive allocations solve the Eisenberg-Gale convex program; hence the outcome is unique and can be approximately {found} by standard gradient methods. An exact algorithm that runs in polynomial time in the number of agents and goods was given by Orlin~\cite{orlin}.
		
		In the case of chores, the competitive rule does not solve any convex optimization problem; instead, competitive allocations correspond to local minima, local maxima, and saddle points of the Nash social welfare on the Pareto frontier of the set of feasible utilities. {The Pareto frontier may contain many such points; consequently, the competitive rule's outcome is no longer unique.}
		
		In this paper, we show that all the outcomes of the competitive rule for chores can be computed in strongly polynomial time if either the number of agents or the number of chores is fixed. The approach is based on a combination of three ideas: all consumption graphs of Pareto optimal allocations can be listed in polynomial time; for a given consumption graph, a {candidate for a competitive {utility profile}} can be constructed via {an} explicit formula; {each candidate can be checked for competitiveness and the allocation can be reconstructed} using a maximum flow computation as in Devanur, Papadimitriou, Saberi, and Vazirani~\cite{devanur2002market}.
		
		Our algorithm immediately gives an approximately-fair allocation of indivisible chores by the rounding technique of Barman and Krishnamurthy~\cite{barman2018}.
\end{abstract}


	\newpage
	{
	\section*{ List of Symbols}
	\paragraph{}
	\renewcommand{\arraystretch}{1.3}

	\begin{tabular}{ll}
		$[n]$ & \textbf{Set of agents:} $[n] = \{1, \ldots, n\}$.  \\ \\ 
		$[m]$ & \textbf{Set of chores:} $[m] = \{1, \ldots, m\}$. \\ \\ 
		$\vec{v}$ & \textbf{Matrix of values:} $v_{i,j}<0$ is the { (negative)} value of agent $i$ for chore $j$. \\ \\ 
		$\vec{z}$ & \begin{tabular}{@{}l@{}l@{}} \textbf{Allocation matrix:} $z_{i,j}\geq 0$ is the fraction  of chore $j$ allocated to agent $i$. \\ An allocation is \emph{feasible} if $\sum_{i=1}^n z_{i,j} =1 $ for each $j \in [m]$. \end{tabular} \\ \\ 
		$\vec{b}$ & \begin{tabular}{@{}l@{}}\textbf{Budget vector:} $b_i<0$ is the budget of agent $i$. \\(Intuition: $|b_i|$ means how much duty  agent $i$ has). \end{tabular} \\ \\ 
		$\vec{u}(\vec{z})$ & \begin{tabular}{@{}l@{}l@{}} \textbf{Utility profile at allocation $\vec{z}$:}\\ $\vec{u}(\vec{z}) = (u_1(\vec{z}_1), \ldots, u_n(\vec{z}_n)),$ where  $u_i(\vec{z}_i) = \sum_{j=1}^{m} v_{i,j}  z_{i,j}$ is agent $i$'s utility\\ for their bundle  $\vec{z}_i = (z_{i,1}, \ldots, z_{i,m})$.\\\end{tabular}\\ \\ 
		$\vec{p}$ & \begin{tabular}{@{}l@{}}\textbf{Price vector:} $p_j<0$ is the price of chore $j$.\\(Intuition: $|p_j|$ is how much is agent paid for doing $j$).\end{tabular}\\ \\ 
		$G_{\vec{z}}$  & \begin{tabular}{@{}l@{}l@{}}\textbf{The consumption graph of allocation $\vec{z}$:} a non-oriented bipartite graph \\ with parts $[n]$ and $[m]$, such that  agent $i \in [n]$ and \\ chore $j \in [m]$  are connected   by  an edge if and only if $z_{i,j}>0$. \end{tabular} \\  \\ 
		$\U(\vec{v})$ & \begin{tabular}{@{}l@{}} \textbf{Set of feasible utility profiles:}\\ $\mathcal{U}(\vec{v})= \{\vec{w} \in\R^n\mid \exists \mbox{ feasible allocation } \vec{z} \mbox{ s.t. }  \vec{w}=\vec{u}(\vec{z})\}\,.$ \end{tabular} \\ \\
		$\U^*(\vec{v})$ & \begin{tabular}{@{}l@{}} \textbf{Set of Pareto optimal utility profiles:}\\ $\U^*(\vec{v})= \{\vec{w} \in\U(\vec{v})\mid \forall \vec{w'} \in\U(\vec{v}) \ (w_i'\geq w_i\ \forall i)\Rightarrow (\vec{w'}=\vec{w})\}\,.$ \end{tabular} \\ \\
		$G_{\tau}(\vec{v})$ & \begin{tabular}{@{}l@{}l@{}} \textbf{Maximal weighted welfare (MWW) graph for weights  $\tau \in \R_{>0}^n$:} \\ $G_{\tau}$ is a   bipartite graph with parts $[n]$ and $[m]$, where agent $i \in [n]$ and \\ chore $j \in [m]$ are linked if $\tau_i \cdot |v_{i,j}|\leq \tau_{i'} \cdot |v_{i',j}|$  for each agent $i' \in [n]$.\end{tabular}\\ \\ 
		$\MWW_{\mbox{non-lonely}}(\vec{v})$ & \begin{tabular}{@{}l@{}} {\textbf{Collection of all MWW graphs with no lonely agents:}}\\
		$\MWW_{\mbox{non-lonely}}(\vec{v})=\{G_\tau(\vec{v})\mid \mbox{$\tau\in \R_{>0}^n$ and  $\deg(i)\ne 0\ \forall i$}\}\,.$
		\end{tabular} \\
	\end{tabular}\\

	}

	\section{Introduction}

	The competitive equilibrium, also known as the market or Walrasian equilibrium,  is a key economic concept that models the allocation of resources at the steady state of an economy when supply equals demand. The economic theory of general equilibrium originated from the ideas of Walras \cite{Walras74} and became mathematically rigorous since the work of Arrow-Debreu \cite{AD54}, who proved the existence of an equilibrium under mild conditions.
	
	Our work is motivated by applications of competitive equilibrium to the problem of fair allocation of resources among agents with different tastes when monetary compensations are not allowed. This extremely fruitful connection between the theory of general equilibrium and economic design was pioneered by Varian \cite{Varian74}. The idea was to give each agent a unit amount of ``virtual'' money and equalize demand and supply: find prices and an allocation such that when each agent spends her money on the most preferred bundles she can afford, all the resources are bought, and all the money is spent. {The division rule that, for each preference profile, outputs the just described allocation is called the competitive equilibrium from equal incomes (CEEI) or the \emph{competitive rule}.}
	
	The resulting allocation has the remarkable fairness property of envy-freeness since all the agents have equal budgets and select their most preferred bundles, and it is Pareto optimal.\footnote{An allocation is Pareto optimal if there is no other allocation in which all the agents are at least as happy and at least one agent is strictly happier.} Due to its desirable properties, the competitive rule has been suggested as a mechanism for allocating goods in real markets, with applications ranging from cloud computing~\cite{DGMVY18} and dividing rent~\cite{spliddit} to assigning courses among university students~\cite{B11}. 
	
	The market considered by Varian is known in the computer science literature as the Fisher market~(named after Irving Fisher, the 19th-century economist, see~\cite{BS00}).
	The properties of the Fisher market were studied in an extensive body of literature, discovering both algorithms for computing equilibria and hardness results (e.g., Chapter 5 of~\cite{AGT_book}).

	For allocating goods, the Fisher market has beautiful structural properties.
	For a large class of utilities,\footnote{Homogeneity and convexity of { utilities} are enough.} the
	equilibria of the Fisher market are captured by the Eisenberg-Gale convex program~\cite{EG}, which maximizes the product of individual utilities.\footnote{The product of utilities is known as the Nash social welfare or the  Nash product from {the axiomatic theory of bargaining}~\cite{Nash50}. Beyond the Fisher market, it balances individual and collective well-being in many settings, e.g., \cite{caragiannis2016unreasonable,conitzer2017fair}.} By the convexity of this problem, the competitive rule is single-valued (i.e., the utilities are unique) and can be computed efficiently for important classes of preferences --- e.g., such as preferences given by additive utilities --- using standard gradient descent methods. 

	\paragraph{Allocation of Bads (Chores).} The literature on resource allocation has largely neglected the study of \emph{bads}, also known as \emph{chores}. These are items that the agents do not want to consume, such as doing housework, or that represent liabilities, such as owing a good to someone. 
	
	\smallskip
	
	While at first sight, the same principles should apply in the problem of allocating chores, as when allocating goods, it turns out that the problem of allocating chores is more complex. The competitive rule for chores was defined and studied in a sequence of papers~\cite{BMSY17,BMSY_SCW} that considered the analog of Fisher markets for chores and analyzed their properties. Even in the case of additive utilities, the competitive rule for chores is no longer single-valued. The equilibrium allocations form a disconnected set and can be obtained as critical points of the {Nash social welfare} on the Pareto frontier of the set of feasible utilities~\cite{BMSY17}. The problem of allocating chores is thus not convex, and the usual techniques for finding competitive allocations based on linear/convex programming, such as primal/dual, ellipsoid, and interior point methods, do not apply.
	
	\smallskip
	
	{For two agents, these obstacles can be easily overcome due to the simple structure of the Pareto frontier: a specific feature of the two-agent case. It implies an efficient procedure for finding competitive allocations, see~\cite{BMSY_SCW} and Section~\ref{subsect_alg_fix_n}.} 
	
	{For more than two agents, computing competitive allocations is mentioned in~\cite{BMSY17} and~\cite{BMSY_SCW} as an open problem. We solve this problem by extending the $2$-agent reasoning of~\cite{BMSY_SCW} to an arbitrary number of agents. This extension is far from straightforward and requires new algorithmic and economic insights that may be interesting on their own.}

	\subsection{Our Contributions} We consider the problem of allocating chores in the basic setting of additive utilities. {It captures situations where the chores are unrelated, i.e., doing one  does not affect the disutility from doing the other, and so the total disutility of an agent is the sum of disutilities from all the chores allocated to this agent.}
	
	\smallskip 
	
	We have a set $[n] = \{1, \ldots, n\}$ of agents and a set $[m] = \{1, \ldots, m\}$ of divisible chores (which may alternatively be seen as indivisible chores that can be allocated in a randomized way). The utilities of the agents are given by a matrix $\vec{v} \in \R_{<0}^{n\times m}$ such that $v_{i,j}$ is the { (negative)} value of agent $i$ for chore $j$. Allocations are defined in the usual way, and the utilities are additive over the allocations. 
	
	\smallskip
	
	We {allow agents to have different budgets} represented by a vector $\vec{b} \in \mathbb{R}_{<0}^n$, which can be seen as virtual currency for acquiring chores; in particular, the budget of an agent denotes how much of a ``duty'' that agent has.\footnote{Most of the literature on fair division assumes that agents are equal in their rights, the case captured by equal budgets. It turns out that allowing unequal budgets is convenient even if in our problem agents have equal rights: see agent-item {parity} in Section~\ref{sect_compute_Pareto} or budget-rounding in Section~\ref{sect_indiv}. } For example, if an agent works full time while another only works half of the time, this can be modeled with budgets $-1$ and $-0.5$, respectively. The budgets are negative to indicate that they represent a liability.
	
	\bigskip

	Our main result is that \emph{{finding all the outcomes of the competitive rule can be solved in strongly polynomial time when either the number of agents or the number of chores is fixed}}. 
	\medskip

	\begin{theorem}[\textbf{Main Theorem, Divisible Chores}]
		Consider a chore division problem $(\vec{v}, \vec{b})$ with $n$ agents and $m$ chores, {where agents have additive utilities given by matrix $\vec{v}\in \R_{<0}^{n\times m}$ and budgets given by vector $\vec{b}\in \R_{<0}^n$.} If $m$ or $n$ are {fixed}, then
		\smallskip
		
		\begin{itemize}
			\item the set of all competitive utility profiles
			\item { a set of competitive allocations and price vectors such that for any competitive utility profile, there is an allocation with this utility profile in the set}
		\end{itemize}
		\smallskip
		
		can be computed in strongly polynomial time, using  $O\left(m^{n(n-1)/{2}+3}\right)$ operations for fixed $n$, or $O\left(n^{m(m-1)/{2}+3}\right)$, for fixed $m$. 
	\end{theorem}
	
	\medskip 
	{The theorem implies an upper bound on the number of competitive equilibria for chores. No such bound was known for $n>2$.}
	
	\medskip 
	
	\textbf{\emph{Overview of the algorithm.}} The theorem is based on the following observation: computing competitive allocations for chores becomes easy if the Pareto frontier is known. Then every face of the frontier is easy to check for containing a competitive allocation. The intuition comes from numerical methods: the solution to a constrained optimization problem is easy to find if we know the active constraints. We show that the Pareto frontier can be computed in polynomial time in the number of chores $m$ for a fixed number of agents $n$ or in polynomial time in $n$ for fixed $m$. This gives a strongly polynomial time algorithm for computing all competitive utility profiles. 
	
	\smallskip
	
	Faces of the Pareto frontier are encoded using the language of consumption graphs. The consumption graph of an allocation is obtained by tracing an edge between an agent and a chore whenever the agent consumes some fraction of that chore.
	Then the first step of the algorithm is to generate {a family of graphs that we call  ``rich'': it must contain} consumption graphs of all Pareto optimal allocations.\footnote{For non-degenerate $\vec{v}$ (all matrices except a subset of measure zero), the set of all ``Pareto optimal'' consumption graphs has polynomial size. To capture the degenerate case, we are forced to define a rich set in a more complex way. {Namely, for each Pareto optimal utility profile, a rich set must contain the consumption graph of \emph{some}  allocation with this profile (but {not necessarily the}  consumption graphs of all such allocations).}} { Such a family contains a graph for each competitive allocation and possibly other graphs that do not correspond to competitive allocations.} In the second step, we generate a ``candidate'' utility profile for each graph in the family by recovering the explicit formula for the utility, assuming that the given graph is a consumption graph of a competitive allocation. In the third step, we adapt the technique from~\cite{devanur2002market} to check if the utility profile considered is, in fact, competitive by studying the amount of flow in an auxiliary maximum flow problem.
	
	\medskip

	\textbf{\emph{Our method and existing techniques}.}
	The {fundamental} difficulty for allocating chores is that {the set of equilibria can be disconnected. The only known algorithmic approach applicable in such a case was introduced by~\cite{DK08} and applied to the case of goods. This approach relies on the so-called cell enumeration technique, a complex tool from computational algebraic geometry~\cite{Basu}. Cell enumeration is used by \cite{DK08} as a black box and the paper mentions finding a simpler construction as an open question.}

	Our algorithm provides the first explicit construction (without cell enumeration) 
	{thus answering the open question from \cite{DK08}. We believe that our approach can also be used for finding competitive allocations in other economies with disconnected sets of equilibria; see Section~\ref{sect_conclusions}.} 
	
	{Our algorithm builds on the observation that the Pareto frontier has a polynomial number of faces, and all of them can be efficiently enumerated. For the $2$-agent case, \cite{BMSY17,BMSY_SCW} use this insight to construct competitive equilibria and mention the algorithm for $n\geq 3$ as an open problem. The two-agent case is special because the Pareto frontier has a simple structure~\cite{aziz2018fair}: items are ordered by the ratio of the values, one agent is allocated the prefix, another agent gets the postfix, and at most one item is split (see Section~\ref{subsect_alg_fix_n}). An extension of this construction to $n\geq 3$, the fact that the number of consumption graphs to inspect remains polynomial in the number of items $m$, and an algorithm enumerating them has not been known.
		These insights are applicable beyond computing the competitive equilibrium: for example, in a follow-up paper~\cite{FedorErel2019}, they play the key role in finding fair allocations with a minimal number of shared items.} 
	
	The idea that we use to compute the Pareto frontier for a fixed number of agents is to recover that of an $n$-agent {problem from the frontiers of ${n (n-1)}/{2}$ auxiliary two-agent sub-problems.} {To show that the resulting family of graphs is rich, we rely on a group of criteria for Pareto optimality (Lemma~\ref{prop_crit_efficiency}). We give concise proofs of these criteria but do not claim priority: these are the folk results, and their analogs are known in different contexts, e.g., cake-cutting~\cite{barbanel2005geometry} (Sections 7,8,10). However, as far as we know, we are the first to harness these criteria for algorithmic purposes. The only exception is the link between weighted welfare and Pareto optimality, which was used by~\cite{echenique2012finding} for approximate computation of competitive equilibria in economies with goods.} 
	
	{When the fixed parameter is the number of chores, the Pareto frontier is obtained using a novel agent-item parity: the set of Pareto optimal consumption graphs turns out to be invariant with respect to switching the roles of agents and chores. We interpret this result as a corollary of the Second welfare theorem (Theorem~\ref{cor_welf}). The Second welfare theorem was not known for chores since the geometric characterization of competitive allocations from~\cite{BMSY_SCW} is obtained under the assumption of equal budgets. We offer a short proof of the characterization with unequal budgets and derive the Second welfare theorem from it.}
	
	{As far as we know, no analogs of 
		explicit formulas recovering the competitive utility profile from the consumption graph were previously known  (Section~\ref{sect_recover_U}). Since we apply these formulas to consumption graphs that may not correspond to competitive allocations, we need to check the competitiveness of the resulting profile. We do that using an auxiliary maximum flow problem, which also recovers the allocation for free. The construction is similar to that of~\cite{devanur2002market} who check the competitiveness of a given vector of prices for goods. However, our case involves a complication: although the flow represents the money spent by agents, the prices are not given, and so we need to define a ``candidate vector of prices'' for a given utility profile.}

	\medskip
	
	\textbf{\emph{{Applications to fair division.}}}
	{A fair division rule is useless in practice if its outcome cannot be computed. For the competitive rule with chores, \emph{no algorithm} (even an inefficient one) was known before our results. In particular,  the definition of competitive allocations is non-constructive and hence does not give any algorithm. The absence of an algorithm precluded using the competitive approach for chore division despite its superior fairness properties compared to other approaches. For example, the popular fair-division platform, \href{http://www.spliddit.org/}{Spliddit.org}, uses the so-called Egalitarian rule to divide chores, which fails to guarantee envy-freeness but can be computed via a simple linear program. Our algorithmic results open the possibility of using the superior competitive rule.}   
	
	{Computing the whole set of competitive equilibria is important since different equilibria may favor different agents. Hence,
		getting access to the whole set} allows one to reason about which outcome to pick for a given instance and to decide on the tradeoffs between different properties, such as maximizing welfare, {additional fairness objectives (e.g., maximizing the minimal utility or making the utilities as equal as possible, maximizing the minimal gap between agent's utility from their allocation and an allocation of other agents), minimizing the number of shared chores} etc., among which there may be tension. Since different equilibria favor different agents, there is a question of picking the best outcome among the set of all equilibria. Bogomolnaia et al.~\cite{BMSY17} suggested picking the median equilibrium in the case of two agents (with probability $1$ there is an odd number of equilibria). The chance that such a selection can be computed without finding the whole set seems negligible. Alternatively, one can select the most egalitarian equilibrium or ask the agents to vote on which equilibrium to select. 
	
	\medskip

	\textbf{\emph{{Implications for} indivisible chores.}}
	{Sometimes, each chore must be allocated to an agent entirely. For example, a manager distributing the tasks may want to avoid responsibility-sharing among workers and hence treats chores as indivisible; alternatively, each chore may require special equipment, and there is only one unit of equipment per chore.} 
	
	{For problems with indivisible chores, we show how to find approximately fair Pareto optimal allocations in polynomial time for fairness notions such as weighted envy-freeness or weighted proportionality. The corresponding relaxations of these fairness notions} are weighted envy-freeness up to the removal of a chore from a bundle and the addition of another chore to another bundle (weighted-EF$_1^1$) and weighted proportionality up to one chore (weighted-Prop$1$). {Neither the existence of Pareto optimal approximately-fair allocations nor algorithms for finding them were known for indivisible chores, e.g., \cite{FrSh19a} mention finding a Pareto optimal Prop$1$ allocation as an open problem.}
	
	{The results on indivisible chores} become an immediate corollary of our main theorem and the technique from the recent work \cite{barman2018} that showed how to round a divisible competitive allocation with goods to get an approximately fair and Pareto optimal indivisible allocation.

	\paragraph{Organization of the paper.} {Section~\ref{sec_preliminaries} defines the competitive rule and discusses its properties, criteria of Pareto optimality, and other tools.} In Section~\ref{sec_main}, we give the main theorem and describe the main phases of the algorithm. {Section~\ref{sect_compute_Pareto} is devoted to computing}  a rich family of {consumption} graphs. Section~\ref{sect_recover_U} provides explicit formulas for competitive utility profiles when the consumption graph is known. In Section~\ref{sect_check_compet}, we check {whether} a given utility profile is competitive, recovering an allocation and prices. {Section~\ref{sect_indiv} is about} approximately fair allocations for indivisible chores. {Section~\ref{sect_conclusions} discusses directions for future research.}

	\subsection{Related Work}\label{subsect_literature}
	\paragraph{Algorithmic results for goods.} The problem of finding polynomial time algorithms for objects defined non-constructively has been a major research focus in the algorithmic game theory literature and beyond \cite{Ppad_inefficient}. Positive results were obtained for important special cases, such as computing Nash equilibria in zero-sum games and competitive equilibria in exchange economies with additive utilities, as well as negative (hardness) results for the corresponding problems in general-sum games and economies with non-additive utilities.
	
	\medskip
	
	\emph{\textbf{The case of ``convex'' economies. }} The search for algorithms for computing competitive equilibria has brought a flurry of efficient algorithms for finding equilibria for {Fisher markets with goods and agents having additive utilities (and for certain extensions of this basic model) 
		as well as computational hardness results beyond the case of additive utilities (e.g., \cite{CSVY,Chen.plc,EY07,LeonFIXP}).}
	
	{All the algorithms for Fisher markets rely on its convexity. Namely, the competitive equilibrium solves the Eisenberg-Gale convex program~\cite{EG}: competitive allocations maximize the {Nash social welfare} (the geometric mean of the utilities weighted by the budgets of the agents).
		This implies the uniqueness of the competitive utility profile and gives polynomial time approximation algorithms based on gradient descent methods (see, e.g., Chapter 6 in~\cite{AGT_book}).}
	
	Surprisingly, {the exact solution to the non-linear Eisenberg-Gale problem is rational and can also be found in polynomial time. The first weakly polynomial algorithm was described by~\cite{devanur2002market}: it relies on the primal-dual approach (see also \cite{AGT_book}, Chapter 5). Using a network-based approach, \cite{orlin} and \cite{vegh2012strongly} constructed strongly polynomial algorithms even if neither $n$ nor $m$ are fixed.}
	
	Similar {results are obtained for the extensions of the Fisher model that preserve the convexity of the problem, such as Eisenberg-Gale and Arrow-Debreu  markets~\cite{jain2007eisenberg, Codenotti05, homothetic,Jain07,garg2019strongly}.
		Other approaches to equilibrium computation include auction-based algorithms \cite{auction.gross} and dynamic procedures such as tatonnement (see, e.g., \cite{CCD13} for a general class containing Eisenberg-Gale markets) and proportional response dynamics for Fisher \cite{Zhang,BDX11,WZ07} and production markets \cite{BMN18}.}
	
	The convexity of the Eisenberg-Gale problem also implies that the equilibrium is robust to small perturbations of the market's parameters~\cite{MV07}. 
	
	In contrast, in the case of chores, {equilibria do not solve any convex optimization problem, there are many of them,}  and no robustness guarantee: {any function that assigns a competitive allocation to each preference profile is necessarily discontinuous~\cite{BMSY17}.  }

	\medskip
	
	\emph{\textbf{{Economies with a disconnected set of equilibria.}}} None of the methods mentioned above applies to the situation when the competitive equilibria {form a disconnected set (as in the case of chores). Such economies often arise when preferences are satiated, or there are constraints on individual consumption\footnote{Note that an economy with chores can be reduced to a constrained economy with goods, see \cite{BMSY17} and Section~\ref{sect_conclusions}.} (see, e.g., \cite{Gjerstad, kirman1986market}).}

	{A survey~\cite{codenotti2004computation}, published in 2004, mentions algorithms for economies with disconnected equilibria as an unexplored research direction.  The work of~\cite{DK08} made the first step in this direction by introducing the algorithmic approach not relying on convexity and applicable to the disconnected case. This approach relies on the black box of the cell enumeration technique discussed above. It was recently applied by~\cite{Tardos}  to the fair assignment problem of~\cite{HZ} (a constrained economy with ``unit-demand buyers'': the total amount of goods allocated to an agent must be equal to $1$).}

	\paragraph{Fair division of an inhomogeneous chore.}  {In the classic model of cake-cutting, agents are dividing an inhomogeneous attractive resource such as pizza with different toppings, land, or time.} This literature typically ignores Pareto optimality focusing exclusively on fairness.

	{With one exception of~\cite{Gardner}, inhomogeneous chores were not considered until recently. \cite{PS02} proposed an envy-free chore division protocol for four agents, and \cite{DFMY18} found a protocol for any number of agents.} \cite{HS15} consider the fair division of chores with connected pieces and bound the loss in social welfare due to fairness. 
	{\cite{segal2018fairly} studied envy-free divisions of a cake that may have some good and some bad parts and showed the existence of connected envy-free allocations for three players. \cite{meunier2018envy} extended this existence result to the case when the number of agents is prime or equal to four.} {The study of the minimal number of queries needed to achieve fairness is initiated by~\cite{farhadi2018complexity}.}

	\paragraph{Relaxed fairness notions for indivisible items.} {If items are indivisible, fair allocations may fail to exist, e.g., when two agents divide one item. The literature on goods has proposed several relaxed notions of fairness applicable in this case:} envy-freeness up to one good (EF1)~\cite{LMMS04}, proportionality up to one good (Prop1), envy-freeness up to any good (EFX), {maximin} fair share~\cite{B11}, and (approximate) competitive equilibrium. {EF1 and Prop1} can be {obtained} by maximizing the {Nash social welfare}, which  also  guarantees Pareto optimality~\cite{CKMP+16}. It is open whether or not EFX allocations always exist (see, e.g., \cite{PR18}).
	
	The {maximin} fair share is a fairness notion inspired by cake-cutting protocols. It requires that each agent's utility be as high as they can guarantee by preparing $n$ bundles and letting the other players choose the best $n-1$ of these bundles. This optimization problem induces the {maximin} value $\alpha_i$ for each player, and the question is whether there exists an allocation where each agent has utility at least $\alpha_i$. While such allocations may not exist~\cite{PW14}, approximations are possible; in particular, there always exists an allocation in which all the agents get two-thirds of their {maximin} value~\cite{PW14}, and this can be computed in polynomial time~\cite{AMNS15}.
	
	{The literature on indivisible chores is sparse.}
	\cite{ARSW17} study the fair allocation of indivisible chores using the {maximin} share solution concept, showing that such allocations do not always exist and computing one (if it exists) is strongly NP-hard; these findings are complemented by a polynomial $2$-approximation algorithm. {\cite{aziz2018fair} consider the problem of fair allocation of a mixture of goods and chores and design several algorithms for finding fair {(but not necessarily Pareto optimal)} allocations in this setting.  \cite{aziz2019strategyproof} enrich the setting by adding the requirement of strategy-proofness and~\cite{bouveret2019chore}, the requirement of connectivity under the assumption that chores form a graph.}
	
	Finally, the competitive rule and its various relaxations (such as those obtained by {partially relaxing} the budget constraint, allowing item bundling, or using randomization) can also be used to allocate indivisible goods. These have been studied for various classes of utilities from the point of view of the existence of fair solutions and their computation in~\cite{B11,FGL13,BHM15,OPR16,BLM16,barman2018,BNT19}. Closest to ours is the work by \cite{barman2018}, which considers Fisher markets with indivisible goods and shows how to compute an allocation that is Prop1 and Pareto optimal in strongly polynomial time. We build on these results to obtain a theorem for chores.

		\medskip
	
	{\textbf{\emph{Follow-up works.}} Since the first draft of our paper was circulated, there has been considerable progress in computing competitive equilibria with chores for various classes of utilities.
	  A version of the Lemke-Howson for a mixture of goods and bads under separable piecewise linear concave utilities was proposed in \cite{chaudhury2021competitive} and demonstrated a good performance in practice. In a model where infinite disutilities are allowed, \cite{chaudhury2022existence} showed that a competitive equilibrium may fail to exist, and checking whether it does is NP-hard. For polynomial approximation algorithms, see ~\cite{boodaghians2022polynomial} and \cite{chaudhury2022competitive}.
	  
	  Computing an exact competitive equilibrium in polynomial time when \emph{both} the number of agents and chores are variable remains an open question for additive utilities.}

	\section{Preliminaries}\label{sec_preliminaries}
	
	There is a set $[n] = \{1, \ldots, n\}$ of agents and a set $[m] = \{1, \ldots, m\}$ of divisible non-disposable chores (bads) to be distributed among the agents. 
	
	\smallskip
	
	A \emph{bundle} of chores is given by a vector $\vec{x} = (x_1, \ldots, x_m)\in \R_{\geq 0}^{m}$, where $x_j$ represents the amount of chore $j$ in the bundle.\footnote{An alternative interpretation for the amount $x_j$ of good $j$ is that it represents the amount of time working on chore $j$ or the probability of getting it.}\footnote{We write $\R_{\geq 0},\R_{>0},\R_{\leq 0},\R_{<0}$ for vectors with non-negative, strictly-positive, non-positive, and strictly negative components, respectively; { to distinguish vectors and scalars bold font is used.}} {Without loss of generality,} there is one unit of each chore.\footnote{Given an arbitrary division problem, one can rescale the utilities to obtain an equivalent problem where the total amount of each chore is one unit.} An \emph{allocation} $\vec{z}=(\vec{z}_i)_{i\in [n]}$ is a set of bundles where agent $i$ receives bundle $\vec{z}_i$. {An allocation is \emph{feasible}} if all the chores are distributed: $\sum_{i=1}^n z_{i,j} =1 $ for each $j \in [m]$. 
	
	\smallskip
	
	The agents have additive utilities specified through a matrix $\vec{v} \in \R_{<0}^{n \times m}$, where $v_{i,j} < 0$ represents the value of agent $i$ for consuming one unit of chore $j$. The utility of agent $i$ in an allocation $\vec{z}$ is $u_i(\vec{z}_i) = \sum_{j =1}^m v_{i,j} \cdot z_{i,j}$. {The \emph{utility profile} of} an allocation $\vec{z}$ is the vector $\vec{u}(\vec{z}) = (u_1(\vec{z}_1), \ldots, u_n(\vec{z}_n))$.
	
	\medskip 
	
	The set of all feasible utility profiles will be denoted by $$\mathcal{U}(\vec{v})= \{\vec{w} \in\R^n\mid \exists \mbox{ feasible allocation } \vec{z} : \ \vec{w}=\vec{u}(\vec{z})\}\,.$$ The set of {feasible} allocations and the set of feasible utility profiles are convex polytopes. 
	
	\medskip
	
	In general, the agents may have different duties with respect to the chores, which will be modeled through different (negative) budgets. Formally, each agent $i$ will be endowed with a budget $b_i < 0$.
	For example, a problem of allocating chores with unequal budgets may arise when a manager assigns tasks to two workers, Alice and Bob. If Alice works full-time and Bob works part-time
	(say 50\%), then it is reasonable that Bob has the right to work half as {much as}   Alice. This corresponds to budgets $b_{Alice} = -1$ and {$b_{Bob} = -0.5$.}
	
	\medskip 
	
	\begin{definition}
		{ A chore division problem $(\vec{v}, \vec{b})$ is a pair of a matrix of values $\vec{v}\in \R_{<0}^{n\times m}$ and budgets $\vec{b}\in \R_{<0}^n$.} 
	\end{definition}

	\subsection{The Competitive Rule}

	Given a vector $\vec{p} = (p_1, \ldots, p_m)\in \R_{<0}^m$, where $p_j$ represents the price of a chore $j$, the price of a bundle $\vec{x}= (x_1, \ldots, x_m)$ of chores is given by
	$p(\vec{x}) = \sum_{j=1}^m p_j \cdot x_j$.
	\bigskip
	
	\begin{definition}[Competitive Allocation]\label{CRdef}
		A {feasible} allocation $\vec{z} = (\vec{z}_1, \ldots, \vec{z}_n)$ for a chore division problem $(\vec{v}, \vec{b})$ with strictly negative matrix of values and budgets is \emph{competitive} if and only if there exists a vector of prices $\vec{p} = (p_1, \ldots, p_m)\in \R_{<0}^m$ such that for each $i \in [n]$:
		\begin{itemize}
			\item Agent $i$'s bundle maximizes her utility among all bundles within her budget; i.e., $u_i(\vec{z}_i) \geq u_i(\vec{x})$ for each bundle $\vec{x}\in \R_{\geq 0}^m$ with $p(\vec{x}) \leq b_i$.
		\end{itemize}
	\end{definition}

	\medskip
	Unlike the Fisher market framework for allocating goods, all the prices and budgets are negative in the case of chores.
Also, Definition~\ref{CRdef} is designed for chores with strictly negative values for all the agents. Chores for which some agents have a zero value can be handled separately as follows.
	
	\bigskip
	
	\begin{remark}[Chores with Zero Utilities]
		Suppose there is at least one chore $j \in [m]$ for which some agent $i$ has a zero value. 
		{Each such chore $j$ can be allocated to the agent $i$ while not affecting the utility of any agent.}
		Moreover, this allocation can be implemented through the competitive rule by setting the price {$p_j=0$.} 
		\qed 
	\end{remark}

	\bigskip
	
	\begin{definition}[Pareto Optimality]
		An allocation $\vec{z}$ is \emph{Pareto optimal} {if it is feasible and there is no other feasible} allocation $\vec{z}'$ in which $u_i(\vec{z}_i')\geq u_i(\vec{z}_i)$ for every agent $i\in [n]$ and the inequality is strict for at least one of them. 
		
		{The utility profile $\vec{u}$ is Pareto optimal if $\vec{u}= \vec{u}(\vec{z})$ for some Pareto optimal allocation $\vec{z}$. The set of all Pareto optimal utility profiles is called the \emph{ Pareto frontier} and is denoted by $\U^*(\vec{v})$.}
	\end{definition}
	
	\bigskip 
	
	\begin{definition}[Weighted Envy-Freeness]
		An allocation $\vec{z}$ is \emph{weighted-envy-free} with weights $\beta\in \R_{>0}^n$  if for every pair of agents $i,j \in [n]$ we have:
		${u_i(\vec{z}_i)}/{\beta_i} \geq {u_i(\vec{z}_{j})}/{\beta_{j}}\,.$
	\end{definition}
	
	\medskip 
	
	The competitive rule satisfies Pareto optimality and weighted-envy-freeness with weights $\beta_i= |b_i|$.
	{To see why weighted-envy-freeness  holds, consider a competitive allocation $\vec{z}$ with budgets $\vec{b}$ and a pair of agents $i$ and $j$. Then the bundle $\vec{x}={|b_i|}/{|b_{j}|} \cdot \vec{z_{j}}$ has  price $p(\vec{x})\leq b_i$ and, hence,   $u_i(\vec{z_i})\geq u_i(\vec{x})$  by  definition of the competitive allocation. This inequality implies envy freeness since $u_i(\vec{x})={|b_i|}/{|b_{j}|} \cdot u_i(\vec{z_{j}})\,.$} For Pareto optimality, see Theorem~\ref{cor_welf}.

	\medskip
	
	\subsection{The geometry of the Competitive Rule and non-convexity}
	In the case of goods, the competitive rule {solves} the \emph{Eisenberg-Gale} optimization problem (see, e.g., Chapter 5 in \cite{AGT_book}): {an allocation $\vec{z}$ is competitive if and only if the \emph{Nash social welfare} $\prod_{i=1}^n \left|u_i(\vec{z}_i)\right|^{|b_i|}$ is maximized at $\vec{z}$.} This problem has a convex programming formulation, and {an approximate solution can be found via} standard gradient descent methods. {The exact strongly polynomial algorithms developed by~\cite{orlin} and~\cite{vegh2012strongly} also rely on convexity (see Related Work).}

	\bigskip
	
	In the case of chores, an analog of the Eisenberg-Gale characterization was found in~\cite{BMSY17}.
	
	\medskip
	
	\begin{theorem}[Bogomolnaia, Moulin, Sandomirskiy, Yanovskaya \cite{BMSY17}]\label{th_geom}
		Consider a chore division problem $(\vec{v}, \vec{b})$.
		A {feasible} allocation $\vec{z}$ is competitive if and only if {its} utility profile $\vec{u}(\vec{z})$ 
		\begin{itemize}
			\item belongs to the set $\U^*(\vec{v})\cap \R_{<0}^n$ of strictly negative points on the Pareto frontier, and 
			\item is a critical point of {Nash social welfare} on the feasible set of utilities $\U(\vec{v})$.
		\end{itemize}
	\end{theorem}
	
	\bigskip

	Recall that a point $x$ is called critical for a smooth function $f$ on a convex set $K$ if the tangent hyperplane to the level curve of $f$ at $x$ is a supporting hyperplane for $K$. {Since the level curve is orthogonal to the gradient $\nabla f(x)$, one gets an equivalent condition:  $x$ is critical if the scalar product $\langle \nabla f(x),\, y-x \rangle$  has the constant sign  for all $y\in K$ (zeros are possible).}
	Local maxima, local minima, and saddle points of $f$ on the boundary of $K$ are examples of critical points.
	
	In~\cite{BMSY17}, the theorem is proved for the case of equal budgets, but the approach extends to arbitrary, strictly negative budgets. A {short} proof is contained in Appendix~\ref{app_sect_characterization} {together with other useful characterizations of competitive allocations.}
	
	\bigskip 
	\begin{remark}
	None of the global extrema of the {Nash social welfare} are competitive for chores: global minima correspond to unfair allocations, where at least one agent receives no chores, and hence the {Nash social welfare} at such allocations is zero; the global maximum lies on the anti-Pareto frontier and therefore it is not Pareto optimal.
		Thus, it is unclear how to use global optimization methods to compute the outcome of the competitive rule. 
		\qed 
	\end{remark}

	\bigskip 
	
	{Throughout the paper, we will use the following examples to illustrate the constructions.}
	
	\bigskip 
	
	\begin{example}\label{ex_1}
		Consider two agents  dividing two chores with values given by
		$$\vec{v}= \left(\begin{array}{cc} -1 & -8 \\  -1 & -2 \end{array}\right),$$
		where rows correspond to agents and columns to chores. Both agents hate the second chore but the second agent finds it less painful compared to the first chore than agent $1$.
		
		We 	{will consider two different cases: equal budgets $\vec{b}=(-3,-3)$ (the agents have identical workload) or unequal budgets  {$\vec{b}=(-2, -4)$} (i.e., the second agent is entitled to twice as much work as agent $1$).} 
		
		{As an application of our algorithm, we will see that for the case of equal budgets, the competitive allocation is unique and is given by
			$$\vec{z}=\left(\begin{array}{cc} 1 & \frac{7}{16} \\  0 & \frac{9}{16} \end{array}\right) \  \mbox{with } \ \vec{p}=\left(-\frac{2}{3},-\frac{16}{3}\right).$$
			This allocation is envy-free: the first agent is indifferent between their own bundle and that of agent~$2$, and the second agent gets a strictly higher utility from their own allocation.}
		
		{For $\vec{b}=(-2, -4)$, we will get two competitive allocations:}
		\begin{itemize}
			\item $\vec{z^1}=\left(\begin{array}{cc} 1 & 0 \\  0 & 1 \end{array}\right), \  \mbox{with prices} \ \vec{p}^1=(-2,-4) $ and 
			\item $\vec{z^2}=\left(\begin{array}{cc} 1 & {1}/{4} \\  0 & {3}/{4} \end{array}\right), \  \mbox{with prices} \ \vec{p}^2=\left(-\frac{2}{3},-\frac{16}{3}\right).$
		\end{itemize}
		These allocations are weighted-envy free: agent $1$ prefers their own allocation to $\frac{1}{2}$ of the allocation of the second agent, while agent~$2$ thinks that their bundle is at least as good as the double bundle of agent~$1$.

		The feasible set $\U(\vec{v})$ and utility profiles of the competitive allocations are depicted in Figure~\ref{fig1} together with the level curves of the {Nash social welfare}: {$|u_1|^3\cdot |u_2|^3$ for equal budgets and $|u_1|^2\cdot |u_2|^4$ for unequal.} 
		The level curves of the {Nash social welfare} (dotted hyperbolas) and the feasible set are not separated by a straight line. Thus, the competitive allocations are not global extrema of the {Nash social welfare}:  the utility profiles {$\vec{u}(\vec{z})$ and $\vec{u}(\vec{z}^2)$ are the local maxima} of the product on the Pareto frontier while the corner of the feasible set, $\vec{u}(\vec{z}^1)$, is a stationary point:  neither a local minimum nor maximum.
		
		This example illustrates that the problem of computing competitive allocations is \emph{non-convex}; there can be \emph{many competitive utility profiles}, and in particular, the set of competitive allocations can be \emph{disconnected}. 
		\begin{figure}[h!]
			\vskip -0.5 cm
			\centering
			{\includegraphics[width=6cm, clip=true, trim=3cm 7.5cm 3cm 6cm]{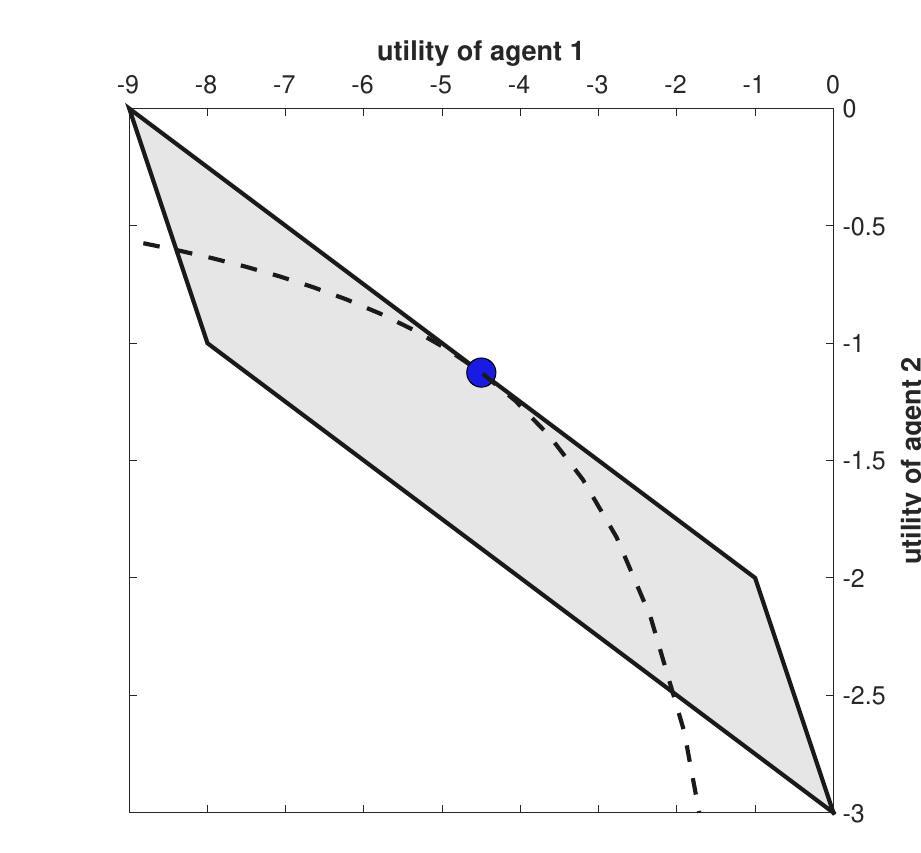}}
			{\includegraphics[width=6cm, clip=true, trim=3cm 7.5cm 3cm 6cm]{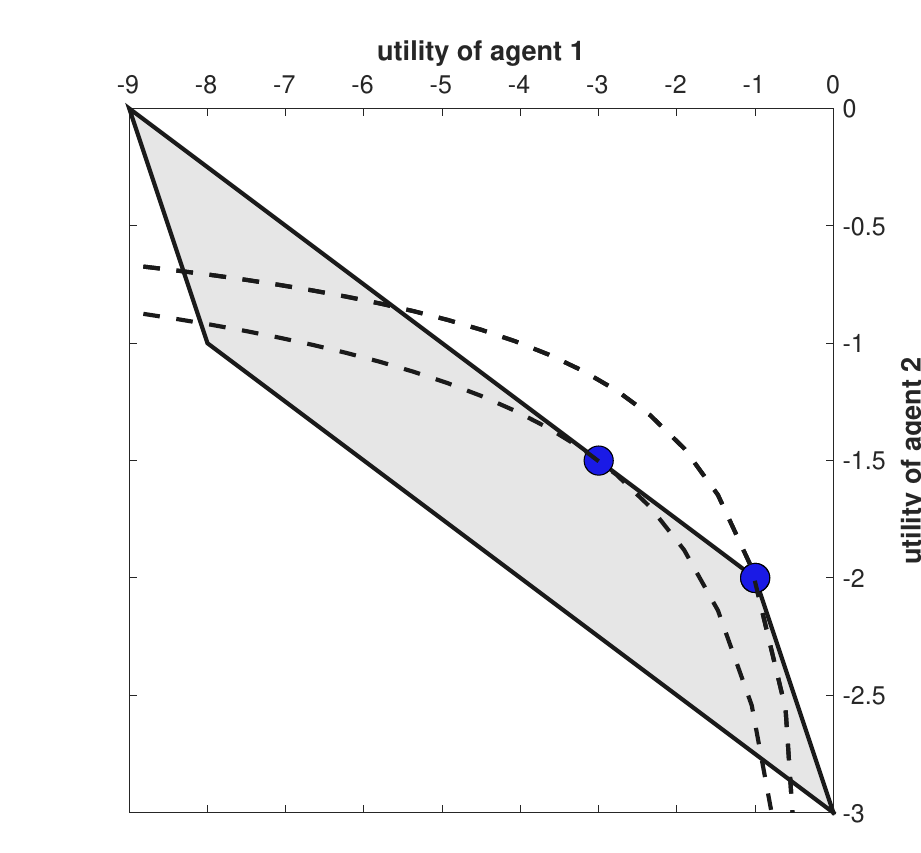}}
			\caption{Competitive utility profiles for Example~\ref{ex_1}. {The left figure depicts the unique {competitive profile} for the case of equal budgets. The right figure is for budgets $\vec{b}=(-2,-4)$:
					the corner corresponds to $\vec{z}^1$ and the profile inside the face is $\vec{u}(\vec{z}^2)$.}}
			\label{fig1}
		\end{figure}
		\qed
	\end{example}
	
	\medskip 
	
	\begin{example}\label{ex_2}
		Consider the {problem with three agents and two chores with the values
			$$\vec{v}= \left(\begin{array}{cc} -1 & -8 \\ -1 & -8\\  -1 & -2 \end{array}\right)$$
			and budgets $\vec{b}=(-1,-1,-4)$.
			
			This problem is obtained by splitting the first agent from Example~\ref{ex_1} (the case of unequal budgets) into two identical agents receiving half of the initial budget each. Hence, the competitive equilibria can be constructed by dividing the bundle of the old first agent into two of equal price and equal value to new agents $1$ and $2$. This leads to the following structure of equilibria:
			$$\vec{z^1}=\left(\begin{array}{cc} \frac{1}{2} & 0\\ \frac{1}{2} & 0 \\  0 & 1 \end{array}\right) \  \mbox{with } \ \vec{p}^1=(-2,-4) \ \mbox{ and } \ \vec{z^2}=t\cdot \left(\begin{array}{cc} 1 & \frac{1}{16} \\
				0 & \frac{3}{16}
				\\ 0 & \frac{3}{4} \end{array}\right) +(1-t) \cdot \left(\begin{array}{cc} 	0 & \frac{3}{16}\\
				1 & \frac{1}{16} 
				\\ 0 & \frac{3}{4} \end{array}\right) \  \mbox{with } \ \vec{p}^2=\left(-\frac{2}{3},-\frac{16}{3}\right),$$
			where $t$ is an arbitrary number in $[0,1]$. 
			
			Instead of one allocation $\vec{z}^2$ from the previous example, we get a continuum of them. This is an artifact caused by the degeneracy of the matrix $\vec{v}$: there is a continuum of ways to split the bundle of the first agent from the previous example into two bundles of the same price and value. The set of degenerate $\vec{v}$ has zero Lebesgue measure; on its complement, each Pareto optimal utility profile corresponds to exactly one allocation (see \cite{BMSY_SCW}, Lemma~1). The formal definition of degeneracy is given in Subsection~\ref{subsect_criteria} below.}
		\qed 
	\end{example}

	\subsection{Corollaries of the geometric characterization. The existence and welfare theorems}
	Theorem~\ref{th_geom}  does not give any recipe for computing outcomes of the competitive rule but allows one to analyze its properties. 
	The first corollary of Theorem~\ref{th_geom} is the existence of competitive allocations. Indeed, there is at least one critical point of the {Nash social welfare} on the Pareto frontier: the one where the level curve of the product, given by the equation $$\left\{\vec{u} \in \R_{<0}^n: \ \prod_{i=1}^n \left|u_i\right|^{|b_i|}=C\right\},$$ first touches the Pareto frontier when we decrease $C$ from large to small values. The corresponding competitive allocation maximizes the {Nash social welfare} \emph{over all Pareto optimal allocations} (see \cite{BMSY17} for details of the construction). 
	
	\medskip
	
	The second corollary of Theorem~\ref{th_geom} is that both welfare theorems hold.
	
	\medskip
	
	\begin{theorem}[Welfare theorems]\label{cor_welf}
		The First and the Second welfare theorems hold:
		\begin{enumerate}
			\item Any competitive allocation is Pareto optimal; 
			\item For any Pareto optimal  allocation $\vec{z}$ with $\vec{u}(\vec{z})\in \R_{<0}^n$, there exist a vector $\vec{b}\in \R_{<0}^n$ such that $\vec{z}$ is competitive for budgets $\vec{b}$.
		\end{enumerate}
	\end{theorem}
	\proof{Proof.}
	By Theorem~\ref{th_geom}, the utility profile of a competitive allocation belongs to $\U^*$, which yields the first item. To prove the second one, note that since $\U$ is a convex polytope and $\vec{u}(\vec{z})$ belongs to its boundary, we can trace a hyperplane $h$ supporting $\U$ at $\vec{u}(\vec{z})$. By Theorem~\ref{th_geom}, it is enough to show that there is a vector of budgets $\vec{b}$, such that $h$ is a tangent hyperplane to the level curve of the {Nash social welfare} at $\vec{u}(\vec{z})$. This condition is satisfied if the gradient of the function 
	$$\ln \left(\prod_{i=1}^n \left|u_i\right|^{|b_i|}\right)$$
	is orthogonal to $h$ at $\vec{u}$. The gradient is {equal to} $(|b_i|/|u_i(\vec{z}_i)|)_{i\in [n]}.$ 
	
	If $h$ is given by the equation  $\left\{V:\ \langle \tau, V\rangle =C\right\}$, then the vector $\tau $ is orthogonal to $h$ and so  it suffices to select $b_i=-|\tau_i| \cdot |u_i(\vec{z}_i)|.$ \qed
	\endproof
	
	\medskip 
	
	{Another} corollary of Theorem~\ref{th_geom} is that whether a given allocation is competitive or not can be determined by its utility profile.
	\medskip
	
	\begin{corollary}[Pareto indifference]\label{corr_Pareto_indiff}
		{If $\vec{u}$ is a competitive utility profile, then any feasible allocation $\vec{z}$ with $\vec{u}=\vec{u}(\vec{z})$ is competitive.}
	\end{corollary} 
	
	{\subsection{Consumption graphs, rich families, and faces of the Pareto frontier}\label{subsect_consumption}
		The \emph{consumption graph} $G_{\vec{z}}$, associated with an allocation $\vec{z}$, 
		is a non-oriented bipartite graph  with parts $[n]$ and $[m]$, where an agent $i \in [n]$ and a chore $j \in [m]$ are connected by an edge if and only if $z_{i,j}>0$. The consumption graph shows who gets what, but 
		does not specify the quantities.}
	
	{The key element of our algorithmic approach is the enumeration of a family of graphs that is rich enough in the sense that each Pareto optimal utility profile has a representative consumption graph.} 
	
	\bigskip
	
	{
		\begin{definition}[Rich family of graphs]\label{def_rich}
			{A collection of bipartite graphs is called \emph{rich} for a given matrix of values  $\vec{v}$ if  
				for any Pareto-optimal utility profile $\vec{u}\in \R_{<0}^n$, 
				there is a feasible allocation $\vec{z}$ with $\vec{u}(\vec{z})=\vec{u}$ such that the consumption graph $G_{\vec{z}}$ belongs to the collection.}
		\end{definition}
	}
	
	\bigskip 
	
	{
		For example, the collection of all bipartite graphs with parts $[n]$ and $[m]$ is rich; however, it contains exponentially many elements.}
	
	{To keep the algorithm polynomial, we will need a rich family of polynomial size. A natural candidate is the set of consumption graphs of all Pareto optimal allocations, which is obviously rich. However, for degenerate problems, this set may also have exponential size; see Remark~\ref{rem_computing_all} showing that the set of consumption graphs corresponding to one particular Pareto optimal utility profile $\vec{u}$ may contain an exponential number of elements.
	}
	
	{
		We modify this idea by considering the set of graphs that correspond to allocations maximizing  \emph{weighted utilitarian welfare} $\sum_{i\in [n]}\tau_i \cdot u_i(\vec{z}_i)$ for some weights $\tau\in \R_{>0}^n$.
	}
	
	\bigskip
	
	{
		\begin{definition}[Maximal Weighted Welfare Graph]\label{def_MWW}
			Let $\tau \in \R_{>0}^n$ be a vector of weights. 
			Consider the $([n],[m])$-bipartite graph where agent $i \in [n]$ and chore $j \in [m]$ are linked if $$\tau_i \cdot |v_{i,j}|\leq \tau_{i'} \cdot |v_{i',j}| \; \; \mbox{ for each agent } i' \in [n].$$
			{In other words, each chore $j$ is connected to all agents $i$ with minimal weighted disutility $\tau_i\cdot |v_{i,j}|$.}

			We call this 
			the \emph{Maximal Weighted Welfare (MWW)} graph and denote it by $G_\tau=G_\tau(\vec{v})$.
		\end{definition} 
	}
	
	\bigskip
	
	\begin{remark}\label{rem_MBB}
		MWW {graphs are related to maximal bang per buck (MBB) graphs introduced by~\cite{devanur2002market} in the case of goods. MBB graphs capture the demand of agents given a price vector $\vec{p}$:  edges are traced between  $i$ and $j$ with the maximal value/price ratio ${v_{i,j}}/{|p_j|}$. 
			
			Comparing the definitions, we see that
			MWW graphs for a problem $\vec{v}$ are related to MBB graphs for $\vec{v}^T$ (the problem where agents and items switched their roles). Namely,
			the MWW graph with weights $\tau$ coincides with the MBB graph for $\vec{v}^T$ with prices $p_j=-{1}/{\tau_j}$.}
		\qed 
	\end{remark}
	
	\bigskip
	
	{We say that a bipartite graph $G$ on $([n],[m])$ has \emph{no lonely agents} if each agent is connected to at least one chore. We denote the collection of all MWW graphs with no lonely agents by $\MWW_{non-lonely}(\vec{v})$.}

	{
		In Section~\ref{sect_compute_Pareto}, we show that the collection $\MWW_{non-lonely}(\vec{v})$ is rich, and its superset can be computed in polynomial time even for degenerate problems if one of the parameters $n$ or $m$ is fixed. 
		Moreover, there is a natural bijection between $\MWW_{non-lonely}(\vec{v})$ and faces of the Pareto frontier.
		Hence enumeration of $\MWW_{non-lonely}(\vec{v})$ can be interpreted as computing the Pareto frontier itself and thus is of independent interest.
	}

	{
		Recall some basics about faces of convex polytopes. Consider the polytope $\U\subset \R^{n}$ of feasible utility profiles. If $h$ is a hyperplane touching the boundary of $\U$, then $\U\cap h$ is a \emph{face} of $\U$, see Figure~\ref{fig2}. This face may have an arbitrary dimension from $0$ (a vertex) to $n-1$ (a proper face of maximal dimension); see \cite{Ziegler} for the introduction to the geometry of polytopes. The Pareto frontier $\U^*$ is a union of faces. 
		\begin{figure}[h!]
			\centering
			{\includegraphics[width=11cm, clip=true, trim=0cm 0cm 0cm 0cm]{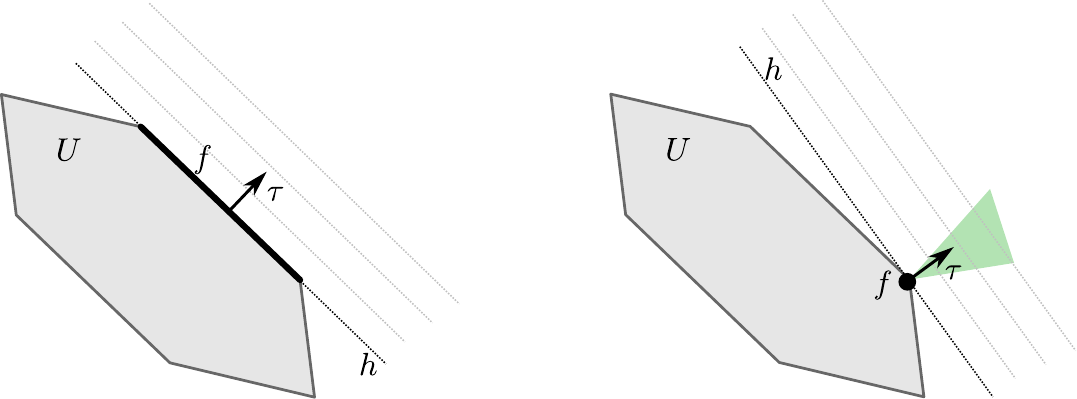}}
			\caption{One-dimensional face $f$ as the set $\vec{u} \in \U$  maximizing $\langle \tau, \vec{u} \rangle$  (left);  a zero-dimensional face, i.e.,  an extreme point (right). For faces $f$ of maximal dimension, the vector  $\tau$  is uniquely defined up to a multiplicative constant  while there is a continuum of $\tau$ for lower-dimensional faces {(the cone in the right figure).}}
			\label{fig2}
		\end{figure}
	}
	
	{
		Assume that the hyperplane $h$ is given by the equation $\left\{ \vec{u} \in \R^n: \ \langle \tau, \vec{u} \rangle = C \right\}$ and fix the sign of $\tau$ in such a way that $\U$ is contained in the half-space $\langle \tau, \vec{x}\rangle \leq C$. Then $f$ has the following dual representation: it maximizes the linear form $\langle \tau, \vec{u} \rangle$ over $\vec{u} \in \U$. The converse is also true: the set of maximizers for any non-zero $\tau$ is a face. 
	}
	
	\subsection{Profitable trading cycles and criteria of Pareto-optimality}\label{subsect_criteria}
	
	\vspace{2mm}
	
	Consider a path in a complete $([n],[m])$-bipartite digraph given by \[ 
	\P=(i_1,j_1,i_2,j_2,\ldots,i_{L}, j_{L},i_{L+1}),\qquad  \text{where}\qquad  L\geq1\,.
	\]
	  We define the product of disutilities along the path as
	\begin{equation}\label{eq_product}
		\pi(\P)=\prod_{k=1}^{L}\frac{|v_{i_k, j_k}|}{|v_{i_{k+1}, j_k}|}\,.
	\end{equation}
	A path $\P$ is a cycle if $i_{L+1}=i_1$; a cycle is \emph{simple} if no agent $i_k$ and no chore $j_k$ enter the cycle twice. 
	
	Consider an allocation $\vec{z}$ and a cycle $\C=(i_1,j_1,\ldots,i_{L+1}=i_1)$, where $L\geq 2$, such that each agent $i_k$ consumes some fraction of chore $j_k$ (i.e., $z_{i_k,j_k}>0$) for $k=1, \ldots,L$. We say that $\C$ is a \emph{profitable trading cycle} for $\vec{z}$ if $\pi(\C)>1$. 
	
	\bigskip
	
	\begin{lemma}\label{prop_crit_efficiency}
		Let $\vec{v} \in \R^{n \times m}_{<0}$ be a matrix of values and $\vec{z}$ be a feasible allocation. Then the following statements are equivalent:
		\begin{enumerate}
			\item The allocation $\vec{z}$ is Pareto optimal.
			\item The allocation {$\vec{z}$ has no} simple profitable trading cycles.\footnote{Similar characterizations of Pareto optimality are known for the house-allocation problem of ~\cite{shapley}, where $n$ indivisible goods (houses) are allocated among $n$ agents with ordinal preferences, one to one. {See~\cite{abdulkadirouglu1998random}} for ex-post efficiency and~\cite{bogomolnaia2001new} for SD-efficiency (aka ordinal efficiency).}
			\item There exists a vector of weights $\tau\in \R_{>0}^n$ such that the consumption graph {of $\vec{z}$} is a subgraph of {the MWW graph $G_\tau(\vec{v})$}.
			\item There exists a vector of weights $\tau\in \R_{>0}^n$ such that the allocation $\vec{z}$ {has the maximal weighted utilitarian welfare {$\sum_{i\in [n]}\tau_i \cdot u_i(\vec{z}_i)$ 
				among all feasible allocations.}}\footnote{The link between Pareto optimal allocations and welfare maximization has a simple geometric origin and holds for any problem with convex set $\U$ of feasible utility profiles. For any point $U$ at the boundary,  we can trace a hyperplane $h$ supporting $\U$. Hence, any $\vec{u}$ on the boundary maximizes  the linear form $\langle \tau, \vec{u'}\rangle$ over $\vec{u'}\in \U$, where $\tau$ is a normal vector to $h$. Thus, the Pareto frontier of $\U$ corresponds to $\tau$ with positive components.}
		\end{enumerate}
	\end{lemma}
	\medskip 
	
	For a more general cake-cutting setting, analogs of items $2$, $3$, and $4$ constitute Sections  8, 10, and 7 of \cite{barbanel2005geometry}. The link between Pareto optimality and weighted utilitarian welfare is classic (see~\cite{Varian74}).
	In contrast to the analogous results from~\cite{barbanel2005geometry}, Lemma~\ref{prop_crit_efficiency} has a short proof (see Appendix~\ref{app_Pareto_criteria}).
	\smallskip

	If the consumption graph {of $\vec{z}$} contains a cycle $\C$ with $\pi(\C)<1$, then by inverting the order of {nodes}, we get a profitable trading cycle. Therefore, the lemma implies the following corollary. 
	
	\bigskip
	
	\begin{corollary}\label{cor_unit_cycles}
		If an allocation $\vec{z}$ is Pareto optimal and {its consumption graph contains a cycle $\C$,} then $\pi(\C)=1$.
	\end{corollary}

	In other words, {the consumption graph of a Pareto optimal allocation} can have cycles only for matrices $\vec{v}$ satisfying certain algebraic equations. This observation is known (see the proof of Lemma~1 in \cite{BMSY_SCW})
	and motivates the following definition.
	\bigskip
	
	{\begin{definition}[Non-degenerate problems]\label{def_non_degenerate}
			We say that the matrix $\vec{v}\in \R_{<0}^{n\times m}$ is \emph{non-degenerate} if for any cycle $\C$ in the complete bipartite graph with parts $[n]$ and $[m]$, the product $\pi(\C)$ is not equal to $1$.
		\end{definition}
		\bigskip
		
		It turns out that non-degenerate problems have better algorithmic properties  (see Remark~\ref{rem_computing_all} { and~\cite{FedorErel2019}}).}
	
	\bigskip
	
	\begin{example}
		Consider {the matrix of values from Example~\ref{ex_1} and show that the equal division $\vec{\overline{z}}$ ($\overline{z}_{i,j}={1}/{2}$ for each agent and chore) is not Pareto optimal. Both agents consume both chores at $\vec{\overline{z}}$, and the cycle $\C=(1,2,2,1)$ has the product $\pi(\C)=\frac{8}{2}\cdot \frac{1}{1}=4$ and, hence, is profitable.  
			
			By Lemma~\ref{prop_crit_efficiency}, the allocation $\vec{\overline{z}}$ is not Pareto optimal. Indeed, consider the following trade along the cycle $\C$: agent~$1$ gives $\varepsilon$ amount of chore~$2$ to agent~$2$ in exchange for $2\varepsilon$ of chore~$1$. This trade is Pareto-improving: agent~$2$ remains indifferent, but the utility of agent $1$ improves by $6\varepsilon$.
			By picking the maximal $\varepsilon$ compatible with feasibility ($\varepsilon={1}/{4}$), we obtain the competitive allocation $\vec{z^2}$ from Example~\ref{ex_1}.}
		
		By {the First welfare theorem, we know that $\vec{z^2}$ is Pareto optimal. But we can also deduce this from Lemma~\ref{prop_crit_efficiency} by guessing the vector $\tau$. Since both agents consume chore $2$, we must have $-8\tau_1=-2\tau_2$. For any such vector $\tau$, the MWW graph $G_\tau$ coincides with the consumption graph of $\vec{z^2}$ and, therefore, $\vec{z^2}$ is Pareto optimal.}
		
		{Matrix $\vec{v}$ from Example~\ref{ex_1} is non-degenerate and the consumption graphs of allocations $\vec{z}$, $\vec{z}^1$, and $\vec{z}^2$ are acyclic, which reflects Corollary~\ref{cor_unit_cycles}. In contrast, the problem from  Example~\ref{ex_2} is degenerate (look at $\C=(1,1,2,2)$). In particular, the consumption graph of the competitive allocation $\vec{z^2}(t)$ for $t={1}/{2}$ contains the cycle $\C=(1,1,2,2)$ despite this allocation being  Pareto optimal by the First welfare theorem.}
		\qed 
	\end{example}

	\section{\textbf{Computing the Competitive Rule for Chores}}\label{sec_main}
	
	In this section, we formulate the main algorithmic result of the paper, discuss its implications and limitations, and present a high-level overview of the algorithm.

	\bigskip
	
	\begin{theorem}\label{main}
		Consider a chore division problem $(\vec{v}, \vec{b})$ with $n$ agents and $m$ chores, {where agents have additive utilities given by matrix $\vec{v}\in \R_{<0}^{n\times m}$ and budgets given by vector $\vec{b}\in \R_{<0}^n$.} If $m$ or $n$ are {fixed}, then
		\smallskip
		
		\begin{itemize}
			\item the set of all competitive utility profiles
			\item { a set of competitive allocations and price vectors such that for any competitive utility profile, there is an allocation with this utility profile in the set}
		\end{itemize}
		\smallskip
		
		can be computed in strongly polynomial time,\footnote{A strongly polynomial algorithm makes a polynomial (in $n$ or  $m$, depending on which of the parameters is fixed) number of elementary operations (multiplication, addition, comparison, etc.). If the input of the problem ($\vec{v}$ and $\vec{b}$) consists of rational numbers in binary representation, then 
			the amount of memory the algorithm uses is bounded by a polynomial in the length of the input. For basics of complexity theory, see~\cite{AB07}.} using  $O\left(m^{n(n-1)/{2}+3}\right)$ operations for fixed $n$, or $O\left(n^{m(m-1)/{2}+3}\right)$, for fixed $m$.
	\end{theorem}
		\bigskip

	\paragraph{\textit{\textbf{What if both $n$ and $m$ are large?}}} Theorem~\ref{main} cannot be improved when both $n$ and $m$ are large. It is known~\cite{BMSY_SCW} that the number of competitive utility profiles can be as large as $2^{\min\{n,m\}}-1$; thus, even listing all competitive utility profiles can take exponential time if both $n$ and $m$ are large. 
	
	Theorem~\ref{main} implies that for bounded  $n$ or $m$,  the number of competitive utility profiles is at most polynomial in the free parameter, which is itself an interesting complement to the exponential lower bound from~\cite{BMSY_SCW}. {Corollary~\ref{cor_number_Pareto} below provides an explicit upper bound: 
		the number of competitive utility profiles is at most} 
	$$\min\left\{\left(2m-1\right)^{\frac{n(n-1)}{2}},\ \left(2n-1\right)^{\frac{m(m-1)}{2}}\right\}.$$
	However, the exponential multiplicity of competitive allocations does not prohibit the existence of an algorithm that finds \emph{one} competitive allocation in polynomial time when both $n$ and $m$ are large.
	
	\smallskip
	
	\begin{open*}
		Is it possible to compute \emph{one competitive utility profile}\footnote{{ Computing one competitive utility profile is equivalent to computing one competitive allocation: given a profile, the corresponding competitive allocation and prices can be found in polynomial time; see Corollary~\ref{cor_complexity_U_is_compet}.}} in time polynomial in $n+m$?\footnote{{As recently shown in \cite{chaudhury2022existence}, the answer is negative in a version of the model with infinite disutilities. In such a model, competitive allocations may fail to exist, and the complexity bottleneck is checking the existence.}}  If such an algorithm exists, it will give a ``computational'' answer to the ``economic'' question posed in~\cite{BMSY17}: finding a single-valued selection of the competitive rule with attractive properties.
	\end{open*}
	
	\paragraph{\textit{\textbf{Computing All Competitive Allocations}}} Theorem~\ref{main} ensures that all competitive utility profiles will be enumerated but does not guarantee to find all the allocations for each such utility profile. 
	It turns out that here the result cannot be improved without restricting the class of preferences {(see Remark~\ref{rem_computing_all}).} 
	
	{For non-degenerate matrices $\vec{v}$  of values (Definition~\ref{def_non_degenerate}), there is only one competitive allocation per utility profile. The algorithm from Theorem~\ref{main} outputs \emph{all competitive allocations}.}
	
	{For degenerate problems, one can have a continuum of competitive allocations with the same utility profile as we saw in Example~\ref{ex_2}. For a given utility profile, the set of competitive allocations is a convex polytope, which may have an exponential number of vertices even if $n$ or $m$ are fixed (Remark~\ref{rem_computing_all}). Thus, for general problems, there is no hope of listing even the set of all extreme points of the set of competitive allocations with a given utility profile.} 
	
	\subsection{The algorithm {and proof of Theorem~\ref{main}}}\label{subsect_algorithm}
	Here we describe a {high-level} structure of the algorithm from Theorem~\ref{main}. The subroutines and underlying ideas are discussed in Sections~\ref{sect_compute_Pareto},~\ref{sect_recover_U}, and~\ref{sect_check_compet}.
	
	{The main ingredient of the algorithm is generating a rich family of graphs (Definition~\ref{def_rich}). Then the algorithm cycles over this family and tries to find a competitive utility profile and an allocation corresponding to each of the graphs; see Algorithm~\ref{algo:main}.}
	
	\SetInd{0.3cm}{0.3cm}
	\SetAlgoHangIndent{0.3cm}
	\SetAlgoSkip{smallskip}
	\SetAlgoInsideSkip{smallskip}
	\setlength{\algomargin}{1em}
	\SetNlSkip{1em}
	
	\medskip
	\SetKwFunction{rich}{Generate\_Rich\_Family\_of\_Graphs}
	\SetKwFunction{candidate}{Compute\_Candidate\_Utility\_Profile}
	\SetKwFunction{iscompetitive}{Is\_Competitive}
	\SetKwFunction{recover}{Recover\_Allocation\_and\_Prices}
	\SetKwFunction{print}{Print}
	{
		\begin{algorithm}[H]
			\DontPrintSemicolon
			\KwIn{values $\vec{v}\in \R_{<0}^{n\times m}$ and budgets $\vec{b}\in \R_{<0}^n$}
			\KwOut{all competitive utility profiles $\vec{u}$ and an allocation-price pair $(\vec{z},\vec{p})$ for each $\vec{u}$}
			$\G=$\rich{$\vec{v}$}\\ \textcolor{cadetgrey}{\tcc{Section~\ref{sect_compute_Pareto} shows how to compute a rich family in polynomial time if either $n$ or $m$ are fixed.\footnote{\vskip -0.4cm {If both $n$ and $m$ are small, one can alternatively define $\G$ to be the set of all bipartite graphs. This leads to a simpler algorithm, which, however, has exponential runtime.}}}} 
			\For {each graph $G\in\mathcal{G}$} {
				$\vec{u}=$\candidate{$G,\vec{v},\vec{b}$} \label{ln_candidate}\\
				\textcolor{cadetgrey}{\tcc{See Section~\ref{sect_recover_U}. If there is a competitive allocation with the consumption graph $G$, the output $\vec{u}$ is its utility profile (in this case, $\vec{u}$ is uniquely defined); if no such allocation exists, $\vec{u}$ is some vector, not necessarily feasible. }}
				\If {\iscompetitive{$\vec{u},\vec{v},\vec{b}$}}{\label{ln_check_compet} 
					$(\vec{z},\vec{p})=$\recover{$\vec{u},\vec{v},\vec{b}$}\\
					\textcolor{cadetgrey}{\tcc{Competitiveness of $\vec{u}$ is checked, and the allocation-price pair can be recovered via a max-flow computation, see Section~\ref{sect_check_compet}.}}
					\print{$\vec{u},\vec{z},\vec{p}$}\\ \textcolor{cadetgrey}{\tcc{Some combinations $\vec{u},\vec{z},\vec{p}$ may be printed several times.}}
				}
			}
			\caption{Computing competitive utility profiles and allocations}
			\label{algo:main}
		\end{algorithm}
	}
	\medskip 
	\proof{Proof of Theorem~\ref{main}.}
	{First, we check the correctness of Algorithm~\ref{algo:main} relying on the correctness of all the subroutines, which is established in the corresponding sections.}
	
	{Since each vector $\vec{u}$ printed by the algorithm is checked for competitiveness (line~\ref{ln_check_compet}), the output of the algorithm is \emph{a subset} of all competitive utility profiles accompanied with allocation-price pairs for each element of this subset. To prove correctness, we must check that the algorithm skips no competitive utility profile.} 
	
	{Let $\vec{u}$ be a competitive utility profile. By Theorem~\ref{th_geom}, the vector $\vec{u}$ belongs to the Pareto frontier and $u_i<0$ for all agents $i\in[n]$. The definition of a rich family implies that there exists a graph $G\in \G$ and a feasible allocation $\vec{z}$ such that $G$ is a consumption graph of $\vec{z}$ and $\vec{u}=\vec{u}(\vec{z})$. By the Pareto indifference property (Corollary~\ref{corr_Pareto_indiff}), $\vec{z}$ is a competitive allocation. We conclude that $G$ corresponds to a competitive allocation. Hence, line~\ref{ln_candidate} of the algorithm will recover the utility profile $\vec{u}$; thus, this profile will not be skipped.} 
	
	{It remains to estimate the time complexity. 
		Since the algorithm cycles over all graphs from the rich family $\G$, the size of $\G$ and the time needed to compute this family determines the overall complexity of the algorithm.} 
	
	{In Section~\ref{sect_compute_Pareto}, we construct $\G$  with the number of elements bounded both by $\left(2m-1\right)^{\frac{n(n-1)}{2}}$ and  by $\left(2n-1\right)^{\frac{m(m-1)}{2}}$ and  show
		that $\G$ can be computed in at most 
		$O\big(m^{\frac{n(n-1)}{2}+1}\big)$ operations for fixed $n$
		or in   $O\big(n^{\frac{m(m-1)}{2}+1}\big)$ for fixed $m$.}
	
	\smallskip 
	
	{By Corollaries~\ref{cor_complexity_recover_U} and~\ref{cor_complexity_U_is_compet}, the time complexity of the for-cycle is bounded by 
		$$|\G|\cdot\left(O\big(nm(n+m)\big)+O\big(n^2m^2(n+m)\big)\right).$$
		Thus the whole algorithm runs in polynomial time, and the number of operations is bounded by $O\left(m^{\frac{n(n-1)}{2}+3}\right)$ for fixed $n$ and by $O\left(n^{\frac{m(m-1)}{2}+3}\right)$ for fixed $m$.} \qed
	\endproof
	
	\section{Rich families of graphs}\label{sect_compute_Pareto}
	{The goal of this section is to construct a rich family of graphs (Definition~\ref{def_rich}) in polynomial time if either $n$ or $m$ is fixed.}

	{We begin with exploring the properties of MWW graphs (Definition~\ref{def_MWW}). We show that the set of MWW graphs encodes faces of the Pareto frontier, is rich, and is invariant with respect to switching the roles of agents and items.}
	
	{Armed with these observations, we design an algorithm enumerating a superset of MWW graphs with no lonely agents. For fixed $n$, the algorithm is built via a reduction to a simple two-agent case. For fixed $m$, we use the aforementioned invariance of MWW graphs.}
	
	\subsection{{Properties of the MWW family}}
	{Lemma~\ref{prop_crit_efficiency} almost implies richness of the set of all MWW-graphs: for any Pareto optimal allocation $\vec{z}$, there is an MWW graph containing the consumption graph of $\vec{z}$ as a \emph{subgraph}. However, the consumption graph itself may not belong to the MWW family.
		Showing richness requires finding a relation between MWW graphs and faces of the Pareto frontier.}
	
	\bigskip
	
	\begin{lemma}\label{lm_faces_and_MWW}
		There is a bijection $f\leftrightarrow G_f$ between faces of the Pareto frontier and {MWW graphs such that the utility profile $\vec{u}(\vec{z})$ of a feasible allocation $\vec{z}$ belongs to a face $f$ if and only if the consumption graph of $\vec{z}$ is a subgraph of $G_f$.} 
	\end{lemma}
	\proof{Proof.}
	There is a one-to-one correspondence between {faces of the Pareto frontier} and solutions to  $\langle\tau, \vec{u} \rangle\to \max $, when $\tau$ ranges over $\R_{>0}^n$ (see Subsection~\ref{subsect_consumption}).
	{The equivalence between statements $(3)$ and $(4)$ of Lemma~\ref{prop_crit_efficiency}} implies the result. \qed
	\endproof
	
	\bigskip
	
	\begin{lemma}\label{lm_efficient_graphs_are_rich_for_PO}
		{The set of all MWW graphs is rich.}
	\end{lemma}
	\proof{Proof.}
	{We need to show that for any Pareto optimal utility profile $\vec{u}\in \R_{<0}^n$, there is a feasible allocation $\vec{z}$ with $\vec{u}=\vec{u}(\vec{z})$ such that the consumption graph of $\vec{z}$ is an MWW graph.}
	
	If $\vec{u}$ is a vertex (i.e., a zero-dimensional face $f$), then we consider {the MWW graph $G_f$} from Lemma~\ref{lm_faces_and_MWW} and pick any feasible allocation $\vec{z}$ with {the consumption graph} $G_f$. Then Lemma~\ref{lm_faces_and_MWW} implies that {$\vec{u}(\vec{z})$ belongs to $f$. Since $f$ consists of only one point, we get $\vec{u}(\vec{z})=\vec{u}$, and we are done.}
	
	If $\vec{u}$ is not a vertex, then we can find a face $f$ of the Pareto frontier such that $\vec{u}$ is in its relative interior  (i.e., $\vec{u} \in f$ but not a boundary point of $f$). Indeed, consider some face $f'$ containing $\vec{u}$; if $\vec{u}$ is not in its relative interior, then we can find a face $f''$ of the boundary of $f'$ such that $\vec{u} \in f''$; since the new face has a smaller dimension, after a finite number of repetitions, we either find the desired face $f$ or find out that $\vec{u}$ is a vertex. 
	
	Fix an auxiliary feasible allocation {$\vec{y}$ with the consumption graph $G_f$.} Then $\vec{u}(\vec{y})\in f$ by Lemma~\ref{lm_faces_and_MWW}. Since $\vec{u}$ belongs to the relative interior, we can represent the utility profile $\vec{u}$ as $$\vec{u} = \varepsilon \cdot \vec{u}(\vec{y})+(1-\varepsilon) \cdot \vec{u}',$$ 
	where the vector $\vec{u}'$ is given by
	$$\vec{u}'=\frac{\vec{u}-\varepsilon \cdot \vec{u}(\vec{y})}{1-\varepsilon}$$ and belongs to $f$ for $\varepsilon>0$ small enough. Consider an allocation $\vec{z}=\varepsilon \cdot \vec{y}+(1-\varepsilon)\vec{z}'$, where {$\vec{z}'$ is a feasible allocation with $\vec{u}(\vec{z'})=\vec{u}'$.} By the construction,  $\vec{u}(\vec{z})=\vec{u}$ and {the consumption graph of $\vec{z}$ coincides with  $G_f$ and thus belongs to the MWW family.} \qed
	\endproof
	\bigskip

	\begin{remark}
		By the construction, the allocation $\vec{z}$ from Lemma~\ref{lm_efficient_graphs_are_rich_for_PO} has the maximal consumption graph with respect to subgraph-inclusion {among all feasible allocations with the utility profile $\vec{u}$.} This gives an alternative interpretation of {the set of MWW} graphs as the set of {maximal consumption graphs corresponding to Pareto optimal utility profiles.}
		\qed 
	\end{remark}
	
	\bigskip
	
	\begin{example}
		{As we see in Figure~\ref{fig1}, the Pareto frontier for the matrix $\vec{v}$ from Example~\ref{ex_1} has two one-dimensional faces. The one having the competitive utility profile $\vec{u}(\vec{z^2})$ in its interior has the form $$t\cdot\footnotesize{\left(\begin{array}{c} -9 \\ 0 \end{array}\right)}+(1-t) \cdot \footnotesize{\left(\begin{array}{c} -1 \\ -2 \end{array}\right)}$$ for $t\in[0,1]$; it is composed of allocations where the first agent gets $(1, t)$ and the second agent gets $(0, 1-t)$. Hence, in the MWW graph corresponding to this face, agent~$1$ is connected to both chores, while the second one is only connected to chore~$2$; this graph originates from the weight vector $\tau=(1,4)$ orthogonal to the face (a vector orthogonal to the face).
			The vertex at the intersection of the two one-dimensional faces (a $0$-dimensional face) is represented by the MWW graph, where agent $1$ is connected to chore $1$ and agent $2$, to chore $2$ (the consumption graph of the allocation $\vec{z^1}$).} 
		
		\smallskip 
		
		{In Example~\ref{ex_1}, the consumption graph of each Pareto optimal allocation is contained in the set of MWW graphs. However, this is no longer true for the degenerate matrix of values $\vec{v}$ from Example~\ref{ex_2}. Consider the common utility profile of the family of allocations $\vec{z^2}(t)$ for $t\in[0,1]$. The consumption graphs for $t=0$ and $t=1$ do not belong to the MWW family. To see this, note that if two agents with identical preferences receive different weights $\tau_i$, then one of them is not connected to any chore in the MWW graph; if, on the other hand, the agents receive equal weights, they are connected to the same set of chores. At $\vec{z^2}$ with $t=0$, agent $1$ does not consume chore $1$, while agent $2$ does, and other way around for $t=1$. Hence, the consumption graphs of these Pareto optimal allocations are not in the MWW family. However, the consumption graph of $\vec{z^2}$ for $t\in(0,1)$ is an MWW graph corresponding to $\tau=(1,1,4)$.} 
		\qed 
	\end{example}
	\paragraph{\textbf{MWW graphs with no lonely agents}}
	{Any rich family remains rich if we exclude those graphs where some agent is not connected to any chore. 
		Indeed, the notion of richness (Definition~\ref{def_rich}) refers to Pareto-optimal utility profiles $\vec{u}$ with \emph{strictly-negative} components and, hence, 
		every agent must consume some chores at any  allocation $\vec{z}$ with $\vec{u}(\vec{z})=\vec{u}$.
	}
	{
		Recall that a graph has \emph{no lonely agents} if each agent is connected to at least one chore and that the set of all MWW graphs with no lonely agents is denoted by $\MWW_{non-lonely}(\vec{v})$.
		
		\smallskip 
		From Lemma \ref{lm_faces_and_MWW} and \ref{lm_efficient_graphs_are_rich_for_PO} we deduce the following corollary.
		
		\medskip
		
		\begin{corollary}\label{cor_MWW_non_lonely_rich}
			The family $\MWW_{non-lonely}(\vec{v})$ is rich. It corresponds to faces $f$ of the Pareto frontier with non-empty intersection $f\cap \R_{<0}^n$ under the bijection of Lemma~\ref{lm_faces_and_MWW}.
		\end{corollary}
	}
	\medskip 
	\paragraph{\textbf{{Agent-item parity of MWW graphs}}} 
	For a matrix of values $\vec{v}$ with $n$ agents and $m$ chores, consider the transposed matrix $\vec{v}^T$ where agents and chores switched their roles, so we have $m$ agents and $n$ chores. There is a natural bijection between bipartite graphs on $([n],[m])$ and $([m],[n])$, and we will not distinguish between them. 
	
	\bigskip
	
	\begin{lemma}\label{prop_agent_items_duality}
		For any $\vec{v}\in \R_{<0}^{n\times m},$ the set of MWW graphs with no lonely agents enjoys the following symmetry
		$$\MWW_{non-lonely}(\vec{v})=\MWW_{non-lonely}(\vec{v}^T).$$
	\end{lemma}
	{This symmetry can be seen as an implication of the Second welfare theorem.  Recall that MWW graphs for $\vec{v}^T$ coincide with the MBB graphs for $\vec{v}$ (Remark~\ref{rem_MBB}). Therefore, the lemma}  
	states that the class of MWW graphs (which {correspond to}  Pareto optimal allocations) coincides with the class of MBB graphs (which {correspond to} competitive allocations). {The formal proof of Lemma~\ref{prop_agent_items_duality} relies on this intuition.}
	\smallskip
	\proof{Proof of Lemma~\ref{prop_agent_items_duality}.}
	By the symmetry of the statement, it is enough to show the inclusion  $$\MWW_{non-lonely}(\vec{v})\subset\MWW_{non-lonely}(\vec{v}^T)\,.$$
	Let $\tau\in \R_{>0}^n$ be an arbitrary vector such that the MWW graph {$G=G_\tau(\vec{v})$ has no lonely agents.  We will show there is a vector $\tau'\in \R_{>0}^m$ such that $G=G_{\tau'}(\vec{v}^T)$.}
	
	Consider a feasible allocation $\vec{z}$ with the consumption graph $G$ for the non-transposed problem. By Lemma~\ref{prop_crit_efficiency}, $\vec{z}$ is Pareto optimal and {all the components of the utility profile $\vec{u}(\vec{z})$ are strictly negative} (here we use the fact that there are no lonely agents). By the Second welfare theorem, $\vec{z}$ is a competitive allocation for some {vector of prices} $\vec{p} \in \R_{<0}^m$ and budgets $b_i=\tau_i \cdot u_i$ (see the proof of Theorem~\ref{cor_welf}). 
	
	Each agent $i$ maximizes their utility on the budget constraint and hence consumes only chores $j$ with the highest ratio ${v_{i,j}}/{|p_j|}$ (Lemma~\ref{lm_MBB}). Equivalently, the consumption graph {of $\vec{z}$ is a subgraph of $G_{\tau'}(\vec{v}^T)$, where $\tau'$ is defined by $\tau'_j={1}/{|p_j|}$ for all $j$. Consequently, $G_\tau(\vec{v})$ is a subgraph of $G_{\tau'}(\vec{v}^T)$.}
	
	Let us show that these two graphs coincide. Assume the converse: there is an edge $(i,j)$ in {$G_{\tau'}(\vec{v}^T)$} that is absent in $G_\tau(\vec{v})$. Let $i'$ be an agent consuming $j$ at $z$. {Hence the edge $(i',j)$ is presented in both $G_\tau(\vec{v})$ and  {$G_{\tau'}(\vec{v}^T)$}.}
	By the definition of {$G_{\tau'}(\vec{v}^T)$,} we have 
	{$$\frac{v_{i,j}}{|p_j|}=\max_{j'\in [m]} \frac{v_{i,j'}}{|p_j'|} \quad\mbox{and}\quad \frac{v_{i',j}}{|p_j|}=\max_{j'\in [m]} \frac{v_{i',j'}}{|p_j'|}.$$
		Since agents spend the whole of their budgets on items with the optimal disutility to price ratio,} we have $v_{i,j} \cdot |b_i| / |p_b| = u_i$  and $v_{i',j} \cdot |b_{i'}| / |p_j| = u_{i'}$. We obtain the identity
	$$\frac{v_{i,j} \cdot b_i}{u_i}=\frac{v_{i',j} \cdot b_{i'}}{u_{i'}}.$$
	Taking into account the relation between $\vec{b}$ and $\tau$, we get $\tau_i \cdot v_{i,j}=\tau_{i'} \cdot v_{i',j}=\max_{i''} \tau_{i''} \cdot v_{i'',j}$, where the last equality follows from the fact that {the edge $(i',j)$ belongs to $G_\tau(\vec{v})$}. Thus, the edge $(i,j)$ must exist in $G_\tau(\vec{v})$. This is a contradiction, which completes the proof.  \qed
	\endproof

	\subsection{{Computing a rich family}}\label{subsect_alg_fix_n}
	{We describe an algorithm enumerating a rich family of graphs $\G$ in polynomial time if either the number of agents $n$ or the number of chores $m$ is fixed. In the case of $n=2$ or $m=2$, the family $\G$ coincides with the family of MWW graphs with no lonely agents and, for general $n$ and $m$, the family $\G$ contains  $\MWW_{non-lonely}$ as a subset.} 
	\medskip
	
	{\begin{lemma}\label{prop_superset}
			For $\vec{v}\in \R_{<0}^{n\times m}$, there is a rich family of graphs $\G(\vec{v})$ with the following properties:
			\medskip 
			\begin{itemize}
				\item The number of graphs in $\G$ is bounded by $\min\Bigl\{(2m-1)^{\frac{n(n-1)}{2}},\, (2n-1)^{\frac{m(m-1)}{2}}\Bigr\}$. \\
				\item Enumerating all the graphs in $\G$ takes  time $O\big(m^{\frac{n(n-1)}{2}+1}\big)$ for fixed $n$ and $O\big(n^{\frac{m(m-1)}{2}+1}\big)$ for fixed $m$.
			\end{itemize}
	\end{lemma}}
	\proof{{Proof for $n=2$.}}
	Reorder all the chores, from those that are relatively harmless to agent $1$ to those that are harmless to agent $2$: the ratios $|v_{1,j}|/|v_{2,j}|$ must be weakly increasing in $j=1,\ldots,m$.  
	Consider the following graphs:
	\medskip 
	
	\begin{itemize}
		\item \emph{$k/(k+1)$-split, for $k=1,\ldots,m-1$ {such that $|v_{1,k}|/|v_{2,k}| < |v_{1,{k+1}}| / |v_{2,{k+1}}|$:}} agent $1$ is linked to all chores $1,\ldots,k$ and agent $2$ is linked to all remaining $k+1,\ldots,m$. No other edges exist.
		\item \emph{$k$-cut, for $k=1,\ldots,m$:} agent $1$ is linked to chores $1,\ldots,k-1$, agent $2$ to chores $k+1,\ldots,m$, and all chores $j$ with  $|v_{1,j}|/|v_{2,j}| = |v_{1,k}|/|v_{2,k}|$ are connected to both agents. No other edges exist.
	\end{itemize}
	\medskip 
	
	\begin{figure}[h!]
		\centering
		\vskip -0.5 cm
		{\includegraphics[width=9cm, clip=true, trim=0cm 0cm 0cm -0.5cm]{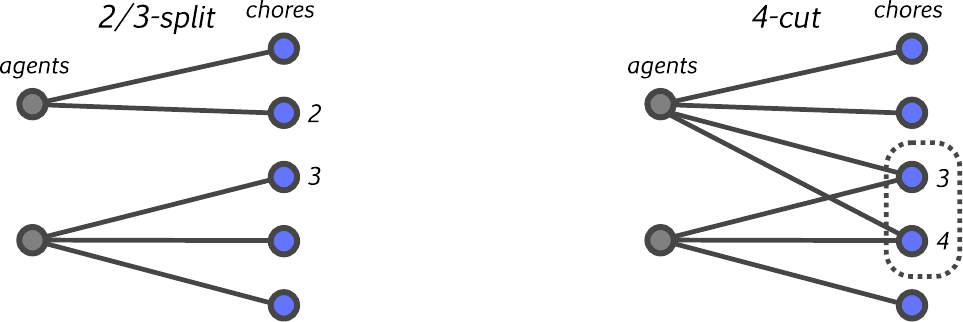}}
		\caption{{Example of $2/3$-split and $4$-cut} for a two-agent problem with $\frac{v_{1,3}}{v_{2,3}}=\frac{v_{1,4}}{v_{2,4}}$. Because of equal ratios, agents share {chores} $3$ and $4$ at the $4$-cut.}
		\label{fig3}
	\end{figure}
	{Let $\G$ be the set of all $k/(k+1)$-splits and $k$-cuts. Then $\G$ has at most $2m-1$ elements, which can be enumerated in time $O(m\cdot\log m)$ (time needed to sort the ratios) or $O(m)$ if chores are already sorted. This implies both items of the lemma for $n=2$.}
	
	{It remains to check that $\G$ is rich. For this purpose, we will demonstrate that $\G$ coincides with the family of MWW graphs with no lonely agents and, hence, is rich by Corollary~\ref{cor_MWW_non_lonely_rich}.
		To show that $ \MWW_{non-lonely}(\vec{v})\subset\G$,} consider an arbitrary graph $G_\tau$ from $\MWW_{non-lonely}(\vec{v})$. Agent $1$ is linked to chores $j$ such that $\tau_1 \cdot |v_{1,j}|\leq \tau_2 \cdot |v_{2,j}|$ or equivalently to those ``prefix'' chores with $$|v_{1,j}|/|v_{2,j}|\leq \tau_2 / \tau_1\,.$$ Similarly, an edge is traced between agent  $2$ and the ``postfix'' for which the following inequality holds: $$|v_{1,j}| / |v_{2,j}|\geq \tau_2 / \tau_1\,.$$
	Thus, if the ratio $\tau_2/\tau_1$ is equal to one of the values $|v_{1,k}| / |v_{2,k}|$, for $k=1,\ldots,m$, we get a split allocation and otherwise a cut. 
	
	{To prove the opposite inclusion $\G\subset \MWW_{non-lonely}(\vec{v})$, we pick $\tau$ such that
		\medskip 
		
		\begin{itemize} 
			\item $|v_{1,k}|/|v_{2,k}|= \tau_2 / \tau_1$  for $k$-cut and \\
			\item  $|v_{1,k}|/|v_{2,k}|< \tau_2 / \tau_1< |v_{1,k+1}|/|v_{2,k+1}|$ for $k/(k+1)$ split.
	\end{itemize}} 
	
	\medskip 
	This completes the proof for $n=2$. \qed 
	\endproof
	\medskip
	
	\proof{{Proof for $3\leq n\leq m$ via a reduction to the two-agent case.}}
	For a division problem with $n$ agents given by a matrix $\vec{v}$,  consider $n(n-1)/2$ auxiliary two-agent problems, where a pair of agents $i \ne i'$ divides the whole set of chores $[m]$ between themselves. {Let $\vec{v}^{\{i,i'\}}$ be the matrix of size  $2\times m$ composed by the two rows $\vec{v_i}$ and $\vec{v_{i'}}$ of matrix $\vec{v}$.} 
	
	\medskip 
	
	{The family of graphs $\G(\vec{v})$ is generated by Algorithm~\ref{algo:enumerate_rich}.} 
	
	\bigskip 
	
	\SetKwFunction{add}{Add}
	{
		\begin{algorithm}[H]
			\DontPrintSemicolon
			\KwIn{values $\vec{v}\in \R_{<0}^{n\times m}$ such that $3\leq n\leq m$}
			\KwOut{a rich family of graphs $\G(\vec{v})$}
			$\G(\vec{v})=\emptyset$ \textcolor{cadetgrey}{\tcc{initialization}}
			precompute $\G(\vec{v}^{\{i,i'\}})$ for each pair $i<i'$\\ 
			\textcolor{cadetgrey}{\tcc{for two-agent sub-problems we already know how to construct $\G$}}
			\For {each combination of graphs $\big(G^{\{i,i'\}}\in \G(\vec{v}^{\{i,i'\}})\big)_{i,i'\in [n], \ i< i'}$} {
				\textbf{construct a graph $G$ as follows:}
				$i\in [n]$ and $j\in[m]$ are connected by the edge if and only if this edge is presented in $G^{\{i,i'\}}$ for all agents $i<i'$\;
				\If {there are no isolated nodes in $G$}{
					Add $G$ to $\G(\vec{v})$\;
				}    
			}
			\Return{$\G$}\;
			\caption{Algorithm enumerating a rich family $\G$ for $3\leq n\leq m$}
			\label{algo:enumerate_rich}
		\end{algorithm}
	}
	
	\medskip 
	
	{Now we check that the constructed family of graphs $\G(\vec{v})$ satisfies the conditions of the lemma.
	
	For each pair of agents, the two-agent family $\G(\vec{v}^{\{i,i'\}})$} contains 
	at most $2m-1$ graphs, thus there are at most $(2m-1)^{\frac{n(n-1)}{2}}$ combinations, {and we obtain the first item of the lemma.
		For fixed $n$, cycling over all combinations requires }  $O\big(m^{\frac{n(n-1)}{2}}\big)$ operations, which for $n\geq 3$ absorbs $O(m\log m)$, the time needed to precompute $\G(\vec{v}^{\{i,i'\}})$ for all pairs of agents $i,i'$. For a given combination of two-agent graphs, $G$ can be constructed using $O(m)$ operations; {this yields the second item of the lemma.}
	
	{To ensure that $\G(\vec{v})$ is rich, we demonstrate that it contains the set $\MWW_{non-lonely}(\vec{v})$, which is rich by Corollary~\ref{cor_MWW_non_lonely_rich}. For each graph $$G_\tau(\vec{v})\in \MWW_{non-lonely}(\vec{v}),$$ we find graphs $G^{\{i,i'\}}$ such that the graph $G$ constructed by Algorithm~\ref{algo:enumerate_rich}  coincides with $G_\tau(\vec{v})$.} Pick $G^{\{i,i'\}}$ equal to the {MWW graph $G_{(\tau_i,\tau_{i'})}$ in the two-agent problem $\vec{v}^{\{i,i'\}}$}. {Since $i$ and $i'$ are connected to some chores in $G_\tau(\vec{v})$,  the graph $G_{(\tau_i,\tau_{i'})}$ has no lonely agents. Hence, $G_{(\tau_i,\tau_{i'})}$ belongs to $\G(\vec{v}^{\{i,i'\}})$.
		By the definition of $G_{(\tau_i,\tau_{i'})}$,} agent $i$ is connected to a chore $j$ in $G^{\{i,i'\}}$ if and only if $$\tau_i \cdot v_{i, j}\geq \tau_{i'} \cdot v_{i',j}\,.$$
	Therefore, an edge $(i,j)$ is traced in $G$ if and only if $\tau_i \cdot v_{i, j}\geq \tau_{i'} \cdot v_{i',j}$ for all $i'$, which, by the definition of MWW graphs, is equivalent to $G=G_\tau(\vec{v})$. {The graph $G$ has no isolated nodes by non-loneliness and hence is added to $\G$ by the algorithm.} 
	This completes the case $3\leq n\leq m$. \qed 
	\endproof
	\medskip
	
	\proof{{Proof for $n > \max\{m,2\}$ via agent-item parity.}}
	{Recall that $\vec{v}^T$ denotes the transposed matrix of values, which corresponds to the {problem obtained by agents and items switching their roles}.} 
	
	{When the number of agents exceeds the number of chores, we define $\G(\vec{v})$ to be equal to $\G(\vec{v}^T)$ and the latter family belongs to the already considered case, where the number of agents is at most the number of chores.}
	
	{Since $\G(\vec{v}^T)$ contains $\MWW_{non-lonely}(\vec{v}^T)$ and MWW graphs satisfy agent-item parity (Lemma~\ref{prop_agent_items_duality}), the set $\G(\vec{v})$ contains $\MWW_{non-lonely}(\vec{v})$ and thus is rich. The estimate on the number of elements and the run time follow directly from already considered cases.}\qed
	\endproof
	\bigskip
	
	\begin{remark} {The constructed rich family $\G$ may contain some redundant elements, namely, consumption graphs of {Pareto dominated} allocations. {Eliminating them may improve the performance of  Algorithm~\ref{algo:main} in practice.} {Graphs of Pareto dominated allocations} can be found using item~2 from Lemma~\ref{prop_crit_efficiency}: {an allocation is Pareto dominated if and only if there exists a} cycle with a  multiplicative weight above $1$ in an auxiliary bipartite graph; such cycles can be detected using, for example, a multiplicative version of the Bellman-Ford algorithm.}
		\qed 
	\end{remark}
	\medskip
	Lemma~\ref{prop_superset} and {Corollary~\ref{cor_MWW_non_lonely_rich}} imply an upper bound on the number of faces of the Pareto frontier. In Section~\ref{sect_recover_U}, we show that there is at most one competitive utility profile per face and, therefore, get the following corollary.
	\medskip
	
	{
		\begin{corollary}\label{cor_number_Pareto}
			The number of faces $f$ of the Pareto frontier with non-empty intersection $f\cap \R_{<0}^n$ and the number of competitive utility profiles (for a given vector of budgets $\vec{b}$) are both at most $$\min\left\{(2m-1)^{\frac{n(n-1)}{2}},\, (2n-1)^{\frac{m(m-1)}{2}}\right\}\,.$$
		\end{corollary}
	}
	
	\bigskip
	
	\begin{example}
		{For the matrix of values from Example~\ref{ex_1}, there are three MWW graphs with non-lonely agents depicted in Figure~\ref{fig_MWW_2_agent}: $1/2$-split, $1$-cut and $2$-cut.}
		
		{For the three-agent matrix $\vec{v}$ from Example~\ref{ex_2}, the set of MWW graphs with no lonely agents can be easily constructed by the agent-item parity. The transposed matrix $\vec{v}^T$ corresponds to a two-agent problem, where the first two chores are identical. For $\vec{v}^T$ there are $3$ different MWW graphs with no lonely agents: $2/3$-split, $1$-cut (coincides with $2$-cut), and $3$-cut. The resulting collection of graphs for $\vec{v}$ is presented in Figure~\ref{fig_MWW_2_agent}.}
		
		\begin{figure}
			\definecolor{olive}{cmyk}{0.21,0,0.56,0.58}
			\begin{center}	
				\begin{tikzpicture}[transform shape,line width=2pt,scale=0.8]
					\foreach \i  in {1,...,2}{%
						\node[draw,inner sep=0.1cm, circle , black!60, fill=black!40] (A-\i) at (0, 3-\i) [thick] {};
						\node at (-0.5, 3-\i) {\i};
						
					}
					
					\foreach \j  in {1,...,2}{%
						\node[draw,inner sep=0.1cm, circle , black!60, fill=blue!70] (C-\j) at (1.5, 3-\j) [thick] {};
						\node at (2, 3-\j) {\j};
					}
					
					\path[black!60] (A-1) edge[-] (C-1);
					\path[black!60] (A-2) edge[-] (C-2);
					
					\node at (0, 2.5) {\emph{agents}};
					\node at (1.5, 2.5) {\emph{chores}};

				\end{tikzpicture}
				\hskip 1cm
				\begin{tikzpicture}[transform shape,line width=2pt,scale=0.8]
					\foreach \i  in {1,...,2}{%
						\node[draw,inner sep=0.1cm, circle , black!60, fill=black!40] (A-\i) at (0, 3-\i) [thick] {};
						\node at (-0.5, 3-\i) {\i};
						
					}
					
					\foreach \j  in {1,...,2}{%
						\node[draw,inner sep=0.1cm, circle , black!60, fill=blue!70] (C-\j) at (1.5, 3-\j) [thick] {};
						\node at (2, 3-\j) {\j};
					}
					
					\path[black!60] (A-1) edge[-] (C-1);
					\path[black!60] (A-2) edge[-] (C-2);
					\path[black!60] (A-2) edge[-] (C-1);
					
					\node at (0, 2.5) {\emph{agents}};
					\node at (1.5, 2.5) {\emph{chores}};
				\end{tikzpicture}
				\hskip 1cm
				\begin{tikzpicture}[transform shape,line width=2pt,scale=0.8]
					\foreach \i  in {1,...,2}{%
						\node[draw,inner sep=0.1cm, circle , black!60, fill=black!40] (A-\i) at (0, 3-\i) [thick] {};
						\node at (-0.5, 3-\i) {\i};
						
					}
					\foreach \j  in {1,...,2}{%
						\node[draw,inner sep=0.1cm, circle , black!60, fill=blue!70] (C-\j) at (1.5, 3-\j) [thick] {};
						\node at (2, 3-\j) {\j};
					}
					
					\path[black!60] (A-1) edge[-] (C-1);
					\path[black!60] (A-2) edge[-] (C-2);
					\path[black!60] (A-1) edge[-] (C-2);
					
					\node at (0, 2.5) {\emph{agents}};
					\node at (1.5, 2.5) {\emph{chores}};
				\end{tikzpicture}
			\end{center}
			\caption{$\G(\vec{v})$ {for  Example~\ref{ex_1}: $1/2$-split, $1$-cut, and $2$-cut.}
				\label{fig_MWW_2_agent}
			}
		\end{figure}
		
		\begin{figure}
			\definecolor{olive}{cmyk}{0.21,0,0.56,0.58}
			\begin{center}
				\hskip 1cm
				\begin{tikzpicture}[transform shape,line width=2pt,scale=0.8]
					\foreach \i  in {1,...,3}{%
						\node[draw,inner sep=0.1cm, circle , black!60, fill=black!40] (A-\i) at (0, 4-\i) [thick] {};
						\node at (-0.5, 4-\i) {\i};
						
					}
					
					\foreach \j  in {1,...,2}{%
						\node[draw,inner sep=0.1cm, circle , black!60, fill=blue!70] (C-\j) at (1.5, 3.5-\j) [thick] {};
						\node at (2, 3.5-\j) {\j};
					}
					
					\path[black!60] (A-1) edge[-] (C-1);
					\path[black!60] (A-2) edge[-] (C-1);
					\path[black!60] (A-3) edge[-] (C-2);
					
					\node at (0, 3.5) {\emph{agents}};
					\node at (1.5, 3) {\emph{chores}};
				\end{tikzpicture}
				\hskip 1cm
				\begin{tikzpicture}[transform shape,line width=2pt,scale=0.8]
					\foreach \i  in {1,...,3}{%
						\node[draw,inner sep=0.1cm, circle , black!60, fill=black!40] (A-\i) at (0, 4-\i) [thick] {};
						\node at (-0.5, 4-\i) {\i};
						
					}
					
					\foreach \j  in {1,...,2}{%
						\node[draw,inner sep=0.1cm, circle , black!60, fill=blue!70] (C-\j) at (1.5, 3.5-\j) [thick] {};
						\node at (2, 3.5-\j) {\j};
					}
					
					\path[black!60] (A-1) edge[-] (C-1);
					\path[black!60] (A-2) edge[-] (C-1);
					\path[black!60] (A-1) edge[-] (C-2);
					\path[black!60] (A-2) edge[-] (C-2);
					\path[black!60] (A-3) edge[-] (C-2);
					
					\node at (0, 3.5) {\emph{agents}};
					\node at (1.5, 3) {\emph{chores}};
				\end{tikzpicture}
				\hskip 1cm
				\begin{tikzpicture}[transform shape,line width=2pt,scale=0.8]
					\foreach \i  in {1,...,3}{%
						\node[draw,inner sep=0.1cm, circle , black!60, fill=black!40] (A-\i) at (0, 4-\i) [thick] {};
						\node at (-0.5, 4-\i) {\i};
						
					}
					
					\foreach \j  in {1,...,2}{%
						\node[draw,inner sep=0.1cm, circle , black!60, fill=blue!70] (C-\j) at (1.5, 3.5-\j) [thick] {};
						\node at (2, 3.5-\j) {\j};
					}
					
					\path[black!60] (A-1) edge[-] (C-1);
					\path[black!60] (A-2) edge[-] (C-1);
					\path[black!60] (A-3) edge[-] (C-2);
					\path[black!60] (A-3) edge[-] (C-1);

					\node at (0, 3.5) {\emph{agents}};
					\node at (1.5, 3) {\emph{chores}};
				\end{tikzpicture}
			\end{center}
			\caption{$\G(\vec{v})$ {for  Example~\ref{ex_2} can be constructed via the agent-item parity. For the transposed problem, the set of MWW graphs with no lonely agents consists of  $2/3$-split, $1$-cut (coincides with $2$-cut), and $3$-cut.}
				\label{fig_MWW_3_agent_ex2}
			}
		\end{figure}

		{In order to illustrate Algorithm~\ref{algo:enumerate_rich}, consider a three-agent problem with three chores:
			\begin{equation}\label{eq_v_3_agents_3_chores}
				\vec{v}=\left(\begin{array}{ccc} -1 & -1 & -1 \\ -8 & -8 & -2\\  -8 & -4 & -1 \end{array}\right).
			\end{equation}
			For the pair of agents $\{1,2\}$, the two-agent sub-problem coincides with the transposed problem of Example~\ref{ex_2} and hence $|\G^{\{1,2\}}|=3$. Similarly, $|\G^{\{2,3\}}|=3$ since chores $2$ and $3$ have the same ratio ${v_{2,j}}/{v_{3,j}}$. For the sub-problem with agents $1$ and $3$, all the ratios ${v_{1,j}}/{v_{3,j}}$ are distinct and hence  $|\G^{\{1,3\}}|=5$. Therefore, the for-cycle of Algorithm~\ref{algo:enumerate_rich} repeats $45$ times (once for each combination of two-agent graphs). 
		
			Figure~\ref{fig_MWW_3_agent_3chores} illustrates how the graph for the original problem is constructed for two particular combinations of two-agent graphs. {Consider the first combination (the top row in the figure). In the two-agent problem with agents $\{1,2\}$, they cut chore $3$.  Hence, in this problem, agent $1$ is connected to the first two chores as \[
			{\max\left\{\frac{|v_{1,1}|}{|v_{2,1}|},\frac{|v_{1,2}|}{|v_{2,2}|}\right\} <\frac{|v_{1,3}|}{|v_{2,3}|}}; 
			\]
see the construction of two-agent MWW graphs from Lemma~\ref{prop_superset}.					{In the two-agent problem with agents $\{1,3\}$, they also cut chore $3$, and the first two chores get connected to the first agent as} \[
		\max\left\{	\frac{|v_{1,1}|}{|v_{3,1}|},\frac{|v_{1,2}|}{|v_{3,2}|} \right\} <\frac{|v_{1,3}|}{|v_{3,3}|}\,.
			\]
		In the two-agent problem with agents $\{2,3\}$, they cut both chores $2$ and $3$ with equal ratios \[
			\frac{|v_{2,2}|}{|v_{3,2}|}=\frac{|v_{2,3}|}{|v_{3,3}|}\,.
			\]
			Moreover, chore~$1$ is connected to agent~$2$ as $\frac{|v_{2,1}|}{|v_{3,1}|}<\frac{|v_{2,3}|}{|v_{3,3}|}$.
			
			\smallskip 
			
			Algorithm~\ref{algo:enumerate_rich} constructs a three-agent MWW graph by connecting  agent $i$ with chore $j$ whenever they are connected in each two-agent graph containing $i$. In the second combination of two-agent MWW graphs (bottom row in the figure), those of agents $\{1,2\}$ and $\{1,3\}$ remain the same, while agents $\{2,3\}$ now cut chore $1$. Hence, chores $2$ and $3$ are connected to agent $3$ as \[
			\max\left\{\frac{|v_{3,2}|}{|v_{2,2}|},\frac{|v_{3,3}|}{|v_{2,3}|} \right\}<\frac{|v_{3,1}|}{|v_{2,1}|}\,.
			\]
			In the resulting three-agent graph, the second agent is lonely, and so this graph is not included in $\G(v)$ by Algorithm~\ref{algo:enumerate_rich}.}} 
		\begin{figure}
			\begin{center}
				\begin{tikzpicture}[transform shape,line width=2pt,scale=0.8]
					\foreach \i  in {1,...,2}{
						\node[draw,inner sep=0.1cm, circle , black!60, fill=black!40] (A-\i) at (0, 4-\i) [thick] {};
						\node at (-0.5, 4-\i) {\i};
						
					}
					
					\foreach \j  in {1,...,3}{
						\node[draw,inner sep=0.1cm, circle , black!60, fill=blue!70] (C-\j) at (1.5, 4-\j) [thick] {};
						\node at (2, 4-\j) {\j};
						
					}
					
					\path[black!60] (A-1) edge[-] (C-1);
					\path[black!60] (A-1) edge[-] (C-2);
					\path[black!60] (A-1) edge[-] (C-3);
					\path[black!60] (A-2) edge[-] (C-3);
					
					\node at (0, 3.5) {\emph{agents}};
					\node at (1.5, 3.5) {\emph{chores}};
					
				\end{tikzpicture}
				\hskip 0.5cm
				\begin{tikzpicture}[transform shape,line width=2pt,scale=0.8]
					\foreach \i  in {1,3}{
						\node[draw,inner sep=0.1cm, circle , black!60, fill=black!40] (A-\i) at (0, 4-\i) [thick] {};
						\node at (-0.5, 4-\i) {\i};
						
					}
					
					\foreach \j  in {1,...,3}{
						\node[draw,inner sep=0.1cm, circle , black!60, fill=blue!70] (C-\j) at (1.5, 4-\j) [thick] {};
						\node at (2, 4-\j) {\j};
						
					}
					
					\path[black!60] (A-1) edge[-] (C-1);
					\path[black!60] (A-1) edge[-] (C-2);
					\path[black!60] (A-1) edge[-] (C-3);
					\path[black!60] (A-3) edge[-] (C-3);

					\node at (0, 3.5) {\emph{agents}};
					\node at (1.5, 3.5) {\emph{chores}};
					
				\end{tikzpicture}
				\hskip 0.5cm
				\begin{tikzpicture}[transform shape,line width=2pt,scale=0.8]
					\foreach \i  in {2,3}{
						\node[draw,inner sep=0.1cm, circle , black!60, fill=black!40] (A-\i) at (0, 4-\i) [thick] {};
						\node at (-0.5, 4-\i) {\i};
						
					}
					
					\foreach \j  in {1,...,3}{
						\node[draw,inner sep=0.1cm, circle , black!60, fill=blue!70] (C-\j) at (1.5, 4-\j) [thick] {};
						\node at (2, 4-\j) {\j};
						
					}
					
					\path[black!60] (A-2) edge[-] (C-1);
					\path[black!60] (A-2) edge[-] (C-2);
					\path[black!60] (A-2) edge[-] (C-3);
					\path[black!60] (A-3) edge[-] (C-2);
					\path[black!60] (A-3) edge[-] (C-3);
					
					\node at (0, 2.5) {\emph{agents}};
					\node at (1.5, 3.5) {\emph{chores}};
					
				\end{tikzpicture}
				\hskip 0 cm
				\begin{tikzpicture}[transform shape,line width=2pt,scale=0.8]
					\node[scale=3]  at (-2, 2.0) {$\to$};
					\foreach \i  in {1,...,3}{
						\node[draw,inner sep=0.1cm, circle , black!60, fill=black!40] (A-\i) at (0, 4-\i) [thick] {};
						\node at (-0.5, 4-\i) {\i};
						
					}
					
					\foreach \j  in {1,...,3}{
						\node[draw,inner sep=0.1cm, circle , black!60, fill=blue!70] (C-\j) at (1.5, 4-\j) [thick] {};
						\node at (2, 4-\j) {\j};
					}
					
					\path[black!60] (A-1) edge[-] (C-2);
					\path[black!60] (A-1) edge[-] (C-3);
					\path[black!60] (A-1) edge[-] (C-1);
					
					\path[black!60] (A-2) edge[-] (C-3);
					
					\path[black!60] (A-3) edge[-] (C-3);
					
					\node at (0, 3.5) {\emph{agents}};
					\node at (1.5, 3.5) {\emph{chores}};
				\end{tikzpicture}

				\vskip 0.3cm
				
				\begin{tikzpicture}[transform shape,line width=2pt,scale=0.8]
					\foreach \i  in {1,...,2}{
						\node[draw,inner sep=0.1cm, circle , black!60, fill=black!40] (A-\i) at (0, 4-\i) [thick] {};
						\node at (-0.5, 4-\i) {\i};
						
					}
					
					\foreach \j  in {1,...,3}{
						\node[draw,inner sep=0.1cm, circle , black!60, fill=blue!70] (C-\j) at (1.5, 4-\j) [thick] {};
						\node at (2, 4-\j) {\j};
						
					}
					
					\path[black!60] (A-1) edge[-] (C-1);
					\path[black!60] (A-1) edge[-] (C-2);
					\path[black!60] (A-1) edge[-] (C-3);
					\path[black!60] (A-2) edge[-] (C-3);
					
					\node at (0, 3.5) {\emph{agents}};
					\node at (1.5, 3.5) {\emph{chores}};
					
				\end{tikzpicture}
				\hskip 0.5cm
				\begin{tikzpicture}[transform shape,line width=2pt,scale=0.8]
					\foreach \i  in {1,3}{
						\node[draw,inner sep=0.1cm, circle , black!60, fill=black!40] (A-\i) at (0, 4-\i) [thick] {};
						\node at (-0.5, 4-\i) {\i};
						
					}
					
					\foreach \j  in {1,...,3}{
						\node[draw,inner sep=0.1cm, circle , black!60, fill=blue!70] (C-\j) at (1.5, 4-\j) [thick] {};
						\node at (2, 4-\j) {\j};
						
					}
					
					\path[black!60] (A-1) edge[-] (C-1);
					\path[black!60] (A-1) edge[-] (C-2);
					\path[black!60] (A-1) edge[-] (C-3);
					\path[black!60] (A-3) edge[-] (C-3);

					\node at (0, 3.5) {\emph{agents}};
					\node at (1.5, 3.5) {\emph{chores}};
					
				\end{tikzpicture}
				\hskip 0.5cm
				\begin{tikzpicture}[transform shape,line width=2pt,scale=0.8]
					\foreach \i  in {2,3}{
						\node[draw,inner sep=0.1cm, circle , black!60, fill=black!40] (A-\i) at (0, 4-\i) [thick] {};
						\node at (-0.5, 4-\i) {\i};
						
					}
					
					\foreach \j  in {1,...,3}{
						\node[draw,inner sep=0.1cm, circle , black!60, fill=blue!70] (C-\j) at (1.5, 4-\j) [thick] {};
						\node at (2, 4-\j) {\j};
						
					}
					
					\path[black!60] (A-2) edge[-] (C-1);
					\path[black!60] (A-3) edge[-] (C-1);
					\path[black!60] (A-3) edge[-] (C-2);
					\path[black!60] (A-3) edge[-] (C-3);
					
					\node at (0, 2.5) {\emph{agents}};
					\node at (1.5, 3.5) {\emph{chores}};
					
				\end{tikzpicture}
				\hskip 0 cm
				\begin{tikzpicture}[transform shape,line width=2pt,scale=0.8]
					\node[scale=3]  at (-2, 2.0) {$\to$};
					\foreach \i  in {1,...,3}{
						\node[draw,inner sep=0.1cm, circle , black!60, fill=black!40] (A-\i) at (0, 4-\i) [thick] {};
						\node at (-0.5, 4-\i) {\i};
						
					}
					
					\foreach \j  in {1,...,3}{
						\node[draw,inner sep=0.1cm, circle , black!60, fill=blue!70] (C-\j) at (1.5, 4-\j) [thick] {};
						\node at (2, 4-\j) {\j};
					}
					
					\path[black!60] (A-1) edge[-] (C-2);
					\path[black!60] (A-1) edge[-] (C-3);
					\path[black!60] (A-1) edge[-] (C-1);

					\path[black!60] (A-3) edge[-] (C-3);
					
					\node at (0, 3.5) {\emph{agents}};
					\node at (1.5, 3.5) {\emph{chores}};
				\end{tikzpicture}

			\end{center}
			\caption{Constructing {$\G(\vec{v})$ for the matrix $\vec{v}$ given by~\eqref{eq_v_3_agents_3_chores} from a combination of three two-agent MWW graphs with no lonely agents (Algorithm~\ref{algo:enumerate_rich}). The first combination leads to an MWW graph for the original problem, which corresponds to weights $\tau=(1,\frac{1}{2},1)$. The second combination results in a graph where agent $2$ is lonely. Hence this graph is not included in $\G(\vec{v})$.}
				\label{fig_MWW_3_agent_3chores}
			}
		\end{figure}
		\qed 
	\end{example}

	\section{Explicit formulas for the competitive utility profile when the consumption graph is known}\label{sect_recover_U}
	
	Suppose we are given an $([n],[m])$-bipartite graph $G$ {with no isolated nodes,} a matrix of values $\vec{v}\in\R_{<0}^{n\times m}$, and a vector of budgets $\vec{b}\in \R_{<0}^n$. Here we derive an explicit formula that expresses the vector $\vec{u}=\vec{u}(\vec{z})$ of utilities for a competitive allocation $\vec{z}$ with budgets $\vec{b}$ and the consumption graph $G_{\vec{z}}=G$ if such an allocation exists. If there is no such $\vec{z}$, the formula {gives} some vector $\vec{u} \in \R_{<0}^n$ { that may not correspond to {any feasible} allocation.}
	The question of recovering $\vec{z}$ from $\vec{u}$ is considered in Section~\ref{sect_check_compet}.
	
	First, we need some notation. Let $N^{i}$ be the set of all agents from the connected component of an agent $i$ in $G$ (including $i$). 
		For $i'\in N_i$ we denote a path connecting $i$ and $i'$ in $G$ by 
		$$\mathcal{P}_{i, i'}=(i=i_1,j_1,i_2,\ldots, i_{L+1}=i')\,.$$
		Recall that $\pi(\mathcal{P})$ is the product of disutilities along the path $\P$ (see Formula~\eqref{eq_product}). We use the convention $\pi(\mathcal{P}_{i, i})=1$. 
		
		\smallskip 
		
		In the following,  we 
denote by
\begin{align} 
\bullet \; \;\;  & \overline{\vec{z}}: \text{the allocation s.t. each chore is divided equally among agents connected to it in $G$;} \notag \\
\bullet \; \; \; & \overline{\vec{u}}: \text{the utility profile of $\overline{\vec{z}}$}\,. \label{eq:notation_auxiliary}
\end{align}
	\bigskip
	
	\begin{lemma}\label{prop_recover_z}
		Fix a division problem $(\vec{v},\vec{b})$ and a graph $G$ {with no isolated nodes.} If there exists a competitive allocation $\vec{z}$ with the consumption graph $G$, the following formula holds for $\vec{u} = \vec{u}(\vec{z})$
		\begin{equation}
			\label{eq_recovering_utilities}
			u_i=\left(\frac{b_i}{\sum_{i'\in N^{i}} b_{i'}}\right) \cdot{\sum_{i'\in N^{i}}  \pi(\mathcal{P}_{i,i'}) \cdot \overline{u}_{i'}},
		\end{equation} 
		where $\overline{\vec{u}}$ and $\overline{\vec{z}}$ are defined in \eqref{eq:notation_auxiliary}.
	\end{lemma}
	\begin{remark}\label{rem_path_independent}
		{If the allocation $\vec{z}$ fulfilling the conditions of the lemma exists, then the right-hand side of formula~\eqref{eq_recovering_utilities} does not depend on the choice of paths $\P_{i,i'}$. Indeed, for any Pareto optimal allocation, the product along a cycle in the consumption graph equals $1$ by Corollary~\ref{cor_unit_cycles}.} 
		
		{If $\vec{z}$ does not exist, the utility profile given by formula~\eqref{eq_recovering_utilities} may depend on the choice of paths and may lead to an infeasible vector $\vec{u} \in \R_{<0}^n$. } 
		\qed 
	\end{remark}
	\medskip 
	
	\proof{Proof of Lemma~\ref{prop_recover_z}.}
	{If two agents $i$ and $i'$ share a chore $j$ at a competitive allocation $\vec{z}$, then their utilities are related (see item $(2)$ of Lemma~\ref{prop_characterizations_of_CA} from the appendix): 
		\begin{equation}\label{eq_shared_chore}
			\frac{v_{i,j} \cdot b_i}{u_i}=\frac{v_{i',j}\cdot b_{i'}}{u_{i'}}.
		\end{equation}
		By iterative application of this equation, we get
		\begin{equation}\label{eq_UiandUj}
			\frac{u_i}{b_i}=\pi(\mathcal{P}_{i,i'}) \cdot \frac{u_{i'}}{b_{i'}}\qquad \mbox{for any agent $i'\in N^i$}.
		\end{equation} 
		Equations~\eqref{eq_UiandUj}} determine utilities for $i'\in N^{i}$ up to a common multiplicative factor. To find it, we need one more condition that relates the components of {$({u}_{i'})_{i'\in N^i}$.} 
	
	Denote by $C^i$ the set of {all} chores consumed by agents from $N^i$.
	These agents spend all their budgets on $C^{i}$ and consume them fully. Therefore, we get the balance equation
	\begin{equation}
		\label{eq_balance}
		\sum_{i'\in N^{i}} b_{i'} = \sum_{j\in C^{i}} p_j.
	\end{equation}
	Using the definition of $\overline{\vec{z}}$ and $\overline{\vec{u}}$ (see \eqref{eq:notation_auxiliary}) and expressing $p_j$ by Lemma~\ref{lm_price_net_utilities}, we rewrite
	the {right-hand side} as 
	$$\sum_{j\in C^{i}} p_j=\sum_{i'\in N^{i}}\sum_{j\in C^i} p_j \cdot \overline{z}_{i',j}=\sum_{i'\in N^{i}}\sum_{j\in C^i} \frac{v_{i',j} \cdot b_{i'}}{u_{i'}} \cdot \overline{z}_{i',j}.$$
	Taking the factor ${b_{i'}}/{u_{i'}}$ out of the internal sum, representing $u_{i'}$ by~\eqref{eq_UiandUj}, and comparing the first equation and the last one, we {obtain}
	$$\sum_{i'\in N^{i}} b_{i'} =\sum_{i'\in N^{i}} \left( \frac{\pi(\mathcal{P}_{i,i'})\cdot b_i}{u_i} \right) \cdot \overline{u}_{i'}.$$
	{Expressing $u_i$ from} this equation, we obtain the required identity \eqref{eq_recovering_utilities}. \qed
	\endproof
	
	\bigskip 
	
	\begin{example}\label{ex_candidates}
		For {the matrix of values $\vec{v}$ from Example~\ref{ex_1}, we have three graphs $G$ to consider (Figure~\ref{fig_MWW_2_agent}): $1/2$-split, $1$-cut, and $2$-cut. Let us compute the candidate utility profile $\vec{u}$ for each of them. The answer depends on the vector of budgets $\vec{b}=(b_1,b_2)$.}
		
		For {$1/2$-split, each agent gets their chore entirely so we obtain $\vec{u}=(-1,-2)$ for any $\vec{b}$.} 
		
		In {order to use the formula~\eqref{eq_recovering_utilities} for $1$-cut, we precompute the following quantities, recalling that $\overline{\vec{u}}$ and $\overline{\vec{z}}$ are defined in \eqref{eq:notation_auxiliary}: $$\pi(\P_{1,1})=\pi(\P_{2,2})=1, \; 
			\pi(\P_{1,2})=\frac{1}{\pi(\P_{2,1})}=\frac{|v_{1,1}|}{|v_{2,1}|}=1,\;  \overline{u}_1=-\frac{1}{2}, \quad \mbox{ and} \quad \overline{u}_2=-\frac{5}{2}\,.$$ We obtain the following utility profile $\vec{u}=\left(\frac{b_1}{b_1+b_2}\cdot (-3),\  \frac{b_2}{b_1+b_2}\cdot (-3)\right)$.}
		
		\medskip 
		
		For {$2$-cut, we do similar preparations: $$\pi(\P_{1,1})=\pi(\P_{2,2})=1, \;  
			\pi(\P_{1,2})=\frac{1}{\pi(\P_{2,1})}=\frac{|v_{1,2}|}{|v_{2,2}|}=4, \; \overline{u}_1=-5, \quad \mbox{and} \quad \overline{u}_2=-1\,.$$ Formula~\eqref{eq_recovering_utilities} gives the utility profile $\vec{u}=\left(\frac{b_1}{b_1+b_2}\cdot (-9),\  \frac{b_2}{b_1+b_2}\cdot \frac{-9}{4}\right)$.}
		
		Thus {we get three candidate utility profiles 
			$$\vec{u}\in \left\{\left(-1,-2\right),\ \left(\frac{-3}{2},\frac{-3}{2}\right),\ \left(\frac{-9}{2},\frac{-9}{8}\right) \right\} \quad\mbox{for $\vec{b}=(-3,-3)$}$$
			and only two profiles
			$$\vec{u}\in \left\{\left(-1,-2\right),\ \left({-3},\frac{-3}{2}\right) \right\}\quad\mbox{for $\vec{b}=(-2,-4)$}$$
			since $1/2$-cut and $1$-cut result in the same utility profile for the latter vector of budgets.}
		\qed 
	\end{example}

	\medskip
	\paragraph{\textbf{Algorithmic consequences.}}
	The following corollary summarises the algorithmic implications of  Lemma~\ref{prop_recover_z}.
	\begin{corollary}\label{cor_complexity_recover_U}
		{ For a given problem $(\vec{v},\vec{b})$ and a graph $G$ with no isolated nodes, the candidate utility profile $\vec{u}$ can be computed using $O(nm(n+m))$ operations even if both $n$ and $m$ are free parameters. }
	\end{corollary}
	\proof{Proof.}
	{If $G$ has a cycle $\C$ with the product $\pi(\C)\ne 1$, then $G$ cannot be a consumption graph of a competitive allocation (Remark~\ref{rem_path_independent}). Such graphs can be identified by the Bellman-Ford algorithm in time $O(nm(n+m))$.}
	
	{If the product is $1$ for any cycle, then the products $\pi(\mathcal{P}_{i,i'})$ do not depend on the choice of paths $\mathcal{P}_{i,i'}$. To compute them for all $i$ and $i'$,} it is enough to find $\pi(\mathcal{P}_{i_0,i'})$ for a fixed $i_0$ in each connected component of $G$ (this runs in time  $O(nm(n+m))$ if the Bellman-Ford algorithm for multiplicative weights is used) and then define $$\pi(\mathcal{P}_{i,i'}) = \frac{\pi(\mathcal{P}_{i_0,i'})}{\pi(\mathcal{P}_{i_0,i})}\,.$$ 
	
	The connected component $N^{i_0}$ of agent $i_0$ can be discovered ``for free'' by the Bellman-Ford algorithm while computing $\pi(\mathcal{P}_{i,i'})$. When all the ingredients are precomputed, finding $u_i$ by formula~\eqref{eq_recovering_utilities} takes $O(n)$ per agent.\qed
	\endproof

	\section{Checking that a given utility profile is competitive. Recovering allocations and prices.}\label{sect_check_compet}
	
	In this section, we consider the following problem. We are given a candidate utility profile $\vec{u}\in\R_{<0}^n$, a vector of budgets $\vec{b} \in \R_{<0}^n$ and a matrix of values $\vec{v}$ without zeros. 
	We do not know whether $\vec{u}$ is feasible or not.
	
	The goal is to check whether $\vec{u}$ can be represented as the utility profile $\vec{u}(\vec{z})$ of a competitive allocation $\vec{z}$ with budget vector $\vec{b}$ and, if the answer is positive, to find at least one such allocation and the corresponding vector of prices. {This goal is achieved using a connection to the maximum flow problem.} 
	
	\subsection{{The maximum flow problem}}\label{subsect_flow} The existence of a competitive allocation $\vec{z}$ with $\vec{u}=\vec{u}(\vec{z}) $ can be checked via a maximum flow problem. {The idea behind the construction is to represent the amount of money spent by each agent on each chore via a flow between them and to define edge capacities to capture budget constraints and the fact that the total spending on each chore cannot exceed its price and that
	each agent spends money only on chores with minimal disutility-to-price ratio. As we will see, the resulting network admits a flow where each agent entirely spends their budget if and only if $\vec{u}$ is a competitive utility profile.} 
	
	{The construction is similar to that of~\cite{devanur2002market}, which checks competitiveness of a given vector of prices. In our case, the vector of prices is not given, and we work with a candidate vector of prices $\vec{q}$ recovered from the candidate utility profile $\vec{u}$ (see below).}

	The network is constructed by adding a source node $s$ and a terminal node $t$ to the complete bipartite graph with parts $[n]$ and $[m]$: the source $s$ is connected to all the agents $[n]$ and the terminal node $t$ is connected to all the chores $[m]$.
	
	The capacity $c$ of each edge is defined as follows. For edges $(s,i)$, we set $c(s,i)=|b_i|$ for each agent $i$. For edges $(i,j)$ between agents and chores, we set $c(i,j)=+\infty$ if this edge exists in the MWW graph $G_{\tau}$ and $c(i,j)=0$ otherwise, where 
	$\tau_i={b_i}/{u_i}, \ i\in [n]\,.$
	For edges from chores to the terminal node, the capacities are defined as $c(j,t)=q_j$, where 
	{$$q_j = \left|\max_{i \in [n]} \; \tau_i \cdot  v_{i,j}\right|\,.$$}
	
	{We have that $\tau$ is proportional to the gradient of the Nash product at $\vec{u}$ and, for competitive allocations, $q_j$ equals the absolute value of the price $p_j$ (Lemma~\ref{lm_price_net_utilities}).}
	
	\bigskip 
	
	\begin{figure}[h!]
		\centering
		\vskip -0.5 cm
		{\includegraphics[width=7.5cm, clip=true, trim=0cm 0cm 0cm 0cm]{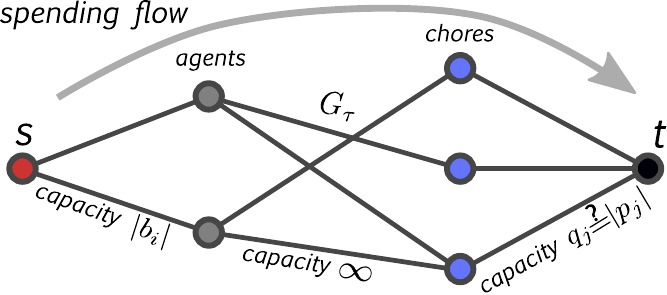}}
		\caption{An example of a network.  Only edges from $G_{\tau}$ are traced between agents and chores, and these edges have a capacity of $\infty$; the capacity of an edge connecting the source $s$  and an agent $i$ is $|b_i|$, and the capacity of an edge from chore $j$ to the terminal node is {$q_j=\left|\max_{i \in [n]} \tau_i  v_{i,j}\right|$.} If there is a competitive {allocation $\vec{z}$}  {such that its consumption graph is a subgraph of} $G_{\tau}$, then $q_j$ is equal to equilibrium price $p_j$ and the flow between agent $i$ and chore $j$ is $|p_j|z_{i,j}$, the absolute amount $i$ spends on $j$.}
		\label{fig4}
	\end{figure}

	\begin{lemma} \label{prop_MaxFlow}
		A utility profile {$\vec{u}\in \R_{<0}^n$} is competitive if and only if the following two  conditions hold:
		\smallskip 
		\begin{itemize}
			\item $\sum_{i\in [n]} |b_i|=\sum_{j\in [m]} q_j$. \\
			\item a maximum flow $\vec{F}$ has the magnitude $\sum_{i\in [n]} |b_i|$. 
		\end{itemize}
		Any such flow defines a competitive allocation $\vec{z}=\vec{z}(\vec{F})$ with $\vec{u} = \vec{u}(\vec{z})$ by $z_{i,j}=F_{i,j}/q_j$ and vice versa.
	\end{lemma} 
	\proof{Proof.}
	Consider a maximum flow $\vec{F}$ of magnitude $\sum_{i\in [n]} |b_i|=\sum_{j\in [m]} q_j$ and check that {$\vec{z}=\vec{z}(\vec{F})$ is competitive.} For all edges $e=(s,i),\  i\in [n]$ and $e=(j,t),\  j\in [m]$, we have $F_e=c(e)$ because the magnitude of $\vec{F}$ equals the capacity of the corresponding cut. Therefore,  we have 
	$$\sum_{i\in [n]} z_{i,j}=\frac{1}{q_j}\sum_{i \in [n]} F_{i,j}= \frac{1}{q_j} F_{j,t}=1$$ 
	for each chore $j$ and 
	hence $\vec{z}$ is a feasible allocation. {Since $F_{i,j}$ is zero on edges $(i,j)$ that do not belong to $G_\tau$, the consumption graph of $\vec{z}$ is a subgraph of $G_\tau$. Hence, each chore $j$ is consumed by an agent $i$ with the highest $\tau_i \cdot v_{i,j}$. This allows us to check that $\vec{u}(\vec{z})=\vec{u}$: 
		\begin{align}
			u_i(\vec{z})&=\sum_{j\in [m]} v_{i,j}\cdot z_{i,j}=\sum_{j\in [m]} \left(\frac{\max_{i'\in[n]} \tau_{i'}\cdot  v_{i',j}}{\tau_i}\right) z_{i,j}=\sum_{j\in [m]} \left(\frac{-q_j}{\tau_i}\right) z_{i,j}\\
			&=\left(\frac{-1}{\tau_i}\right) \cdot \sum_{j\in [m]} F_{i,j}=\left(\frac{-1}{\tau_i}\right) \cdot F_{s,i}= \left( \frac{-u_i}{b_i}\right) \cdot|b_i|=u_i.
		\end{align}
		Since $\vec{z}$ is a feasible allocation, we conclude that $\vec{u}(\vec{z})=\vec{u}$ is a feasible utility profile. The consumption graph of $\vec{z}$ is a subgraph of $G_\tau$ and, hence, Lemma~\ref{prop_crit_efficiency} implies that $\vec{z}$ maximizes the weighted welfare $\sum_{i}\tau_i u_i(\vec{z}_i)$ over feasible allocations. Therefore, $\vec{u}$ maximizes $\sum_{i}\tau_i  u_i$ over feasible utilities. Thus $\vec{u}$ belongs to the Pareto frontier and the hyperplane $$\left\{\vec{u'}\in \R^n\,:\, \sum_{i}\tau_i  u_i=\sum_{i}\tau_i  u_i' \right\}$$ supports the set of feasible utilities at $\vec{u}$. Since $\tau$ is proportional to the gradient of the Nash product $\prod_{i=1}^n |u_i|^{|b_i|}$, the profile $\vec{u}$ is a critical point of the Nash product. By Theorem~\ref{th_geom}, we have that $\vec{z}$ is a competitive allocation and $\vec{u}$ is a competitive utility profile.}

	Now we check the opposite direction: if $\vec{z}$ is a competitive allocation and $\vec{u}=\vec{u}(\vec{z})$, then $$\sum_{i\in [n]} |b_i|=\sum_{j\in [m]} q_j$$ and there is a maximum flow $\vec{F}$ with magnitude $\sum_{i\in [n]} |b_i|$. Indeed, as $\vec{z}$ is competitive, we have $q_j=-p_j$, where $\vec{p}$ is the competitive price vector (Lemma~\ref{lm_price_net_utilities}). Then 
	the condition $$\sum_{i\in [n]} |b_i|=\sum_{j\in [m]} q_j$$ is satisfied because the amount of money spent equals the sum of prices. Consider the flow $\vec{F}$ that represents how much money (in absolute value) each agent $i$ spends on a particular chore $j$. Define $F_{s,i}$ as the total spending of $i$ which equals $|b_i|$
	(and thus the flow saturates each edge $(s,i)$ and has the proper magnitude); $F_{i,j}=z_{i,j}|p_j|$ is the absolute value  of the amount  $i$ spends on $j$. Thus we have the balance equation $$F_{s,i}=|b_i|=\sum_{j\in [m]} F_{i,j}\,.$$ {By Theorem~\ref{th_geom}, we have that $\vec{u}$ is a critical point of the Nash product and, hence, $\vec{z}$ maximizes $\sum_{i} \tau_i u_i(\vec{z}_i)$ (recall that $\tau$ is proportional to the gradient of the Nash product at $\vec{u}$). From this we deduce that agent $i$ consumes $j$ only if $\tau_i v_{i,j}$ is maximal among all agents or, equivalently, that 
		the consumption graph of $\vec{z}$ is a subgraph of $G_\tau$. Thus capacity constraints for edges $(i,j)$ are satisfied.} For each chore $j$ the inflow $\sum_{i}F_{i,j}$ is the total number of money spent on $j$, which equals the absolute value of the  price $|p_j|$. Therefore defining $F_{j,t}=|p_j|$ we get a feasible flow of magnitude $\sum_{i\in [n]} |b_i|$, which completes the proof. \qed
	\endproof	
	
	\bigskip

	\begin{example}\label{ex_max_flow}
		We construct the corresponding network for {each of the candidate utility profiles precomputed in Example~\ref{ex_candidates}. Figure~\ref{fig_Networks_2_agent_equal_budgets} depicts the networks and edge capacities for the case of equal budgets $\vec{b}=(-3,-3)$ and Figure~\ref{fig_Networks_2_agent_unequal_budgets} represents $\vec{b}=(-2,-4)$. For equal budgets, only the third network corresponding to the utility profile $$\left(-\frac{9}{2},-\frac{9}{8}\right)$$ admits a flow of magnitude $b_1+b_2=6$. We deduce that there is only one competitive utility profile in this case.  The {maximum flow} leads to the competitive allocation $\vec{z}$ from Example~\ref{ex_1}.
		
		For $\vec{b}=(-2,-4)$, both networks admit a flow of magnitude $6$. Thus both utility profiles are competitive; the maximum flows give allocations $\vec{z^1}$ and $\vec{z}^2$ from  Example~\ref{ex_1}.} \qed 
		\begin{figure}
			\definecolor{olive}{cmyk}{0.21,0,0.56,0.58}
			\begin{center}	
				\begin{tikzpicture}[transform shape,line width=2pt,scale=0.8]

					\foreach \i  in {1,...,2}{%
						\node[draw,inner sep=0.1cm, circle , black!60, fill=black!40] (A-\i) at (0, 3-\i) [thick] {};
					}
					
					\foreach \j  in {1,...,2}{%
						\node[draw,inner sep=0.1cm, circle , black!60, fill=blue!70] (C-\j) at (1.5, 3-\j) [thick] {};
					}
					
					\path[black!60] (A-1) edge[-] (C-1);
					\path[black!60] (A-2) edge[-] (C-2);

					\node[draw,inner sep=0.1cm, circle , black!60, fill=red!70] (S) at (-1.5, 1.5) [thick] {};	
					\node[draw,inner sep=0.1cm, circle , black!60, fill=black] (T) at (3, 1.5) [thick] {};
					\node at (-1.5, 1.9) {$s$};
					\node at (3, 1.9) {$t$};

					\path[black!60] (S) edge[-] (A-1);
					\path[black!60] (S) edge[-] (A-2);
					\path[black!60] (C-1) edge[-] (T);
					\path[black!60] (C-2) edge[-] (T);
					
					\node at (-0.7, 2) {$3$};
					\node at (-0.7, 1) {$3$};
					\node at (0.75, 2.2) {$\infty$};
					\node at (0.75, 0.8) {$\infty$};
					\node at (2.2, 2) {$2$};
					\node at (2.2, 1) {$4$};
					\node[scale=1.2] at (-0.5, 2.8) {$\vec{u}=\left(-1,-2\right):$};

				\end{tikzpicture}
				\hskip 1cm
				\begin{tikzpicture}[transform shape,line width=2pt,scale=0.8]
					
					\foreach \i  in {1,...,2}{
						\node[draw,inner sep=0.1cm, circle , black!60, fill=black!40] (A-\i) at (0, 3-\i) [thick] {};
					}
					
					\foreach \j  in {1,...,2}{
						\node[draw,inner sep=0.1cm, circle , black!60, fill=blue!70] (C-\j) at (1.5, 3-\j) [thick] {};
					}
					
					\path[black!60] (A-1) edge[-] (C-1);
					\path[black!60] (A-2) edge[-] (C-2);
					\path[black!60] (A-2) edge[-] (C-1);

					\node[draw,inner sep=0.1cm, circle , black!60, fill=red!70] (S) at (-1.5, 1.5) [thick] {};	
					\node[draw,inner sep=0.1cm, circle , black!60, fill=black] (T) at (3, 1.5) [thick] {};
					\node at (-1.5, 1.9) {$s$};
					\node at (3, 1.9) {$t$};

					\path[black!60] (S) edge[-] (A-1);
					\path[black!60] (S) edge[-] (A-2);
					\path[black!60] (C-1) edge[-] (T);
					\path[black!60] (C-2) edge[-] (T);
					
					\node at (-0.7, 2) {$3$};
					\node at (-0.7, 1) {$3$};
					\node at (0.75, 2.2) {$\infty$};
					\node at (0.75, 0.8) {$\infty$};
					\node at (2.2, 2.1) {$3/2$};
					\node at (2.2, 1) {$3$};
					\node[shift={(0.60, 1.62)}, rotate=35] {$\infty$};
					\node[scale=1.2] at (-0.5, 2.8) {$\vec{u}=\left(-\frac{3}{2},-\frac{3}{2}\right):$};
				\end{tikzpicture}
				\hskip 1cm
				\begin{tikzpicture}[transform shape,line width=2pt,scale=0.8]

					\foreach \i  in {1,...,2}{%
						\node[draw,inner sep=0.1cm, circle , black!60, fill=black!40] (A-\i) at (0, 3-\i) [thick] {};
					}
					
					\foreach \j  in {1,...,2}{%
						\node[draw,inner sep=0.1cm, circle , black!60, fill=blue!70] (C-\j) at (1.5, 3-\j) [thick] {};
					}
					
					\path[black!60] (A-1) edge[-] (C-1);
					\path[black!60] (A-2) edge[-] (C-2);
					\path[black!60] (A-1) edge[-] (C-2);

					\node[draw,inner sep=0.1cm, circle , black!60, fill=red!70] (S) at (-1.5, 1.5) [thick] {};	
					\node[draw,inner sep=0.1cm, circle , black!60, fill=black] (T) at (3, 1.5) [thick] {};
					\node at (-1.5, 1.9) {$s$};
					\node at (3, 1.9) {$t$};

					\path[black!60] (S) edge[-] (A-1);
					\path[black!60] (S) edge[-] (A-2);
					\path[black!60] (C-1) edge[-] (T);
					\path[black!60] (C-2) edge[-] (T);
					
					\node at (-0.7, 2) {$3$};
					\node at (-0.7, 1) {$3$};
					\node at (0.75, 2.2) {$\infty$};
					\node at (0.75, 0.8) {$\infty$};
					\node at (2.2, 2.1) {$2/3$};
					\node at (2.2, 0.8) {$16/3$};
					\node[shift={(0.60, 1.38)}, rotate=-35] {$\infty$};
					\node[scale=1.2] at (-0.5, 2.8) {$\vec{u}=\left(-\frac{9}{2},-\frac{9}{8}\right):$};
				\end{tikzpicture}
			\end{center}
			\caption{The {networks and capacities for Example~\ref{ex_1} with budgets $\vec{b}=(-3,-3)$ and utility profiles computed in Example~\ref{ex_candidates}. The third network is the only one admitting a flow of magnitude of $6$ (the sum of the budgets). Hence, only the third utility profile is competitive.}
				\label{fig_Networks_2_agent_equal_budgets}
			}
		\end{figure}
		
		\begin{figure}
			\definecolor{olive}{cmyk}{0.21,0,0.56,0.58}
			\begin{center}	
				\begin{tikzpicture}[transform shape,line width=2pt,scale=0.8]
					
					\foreach \i  in {1,...,2}{
						\node[draw,inner sep=0.1cm, circle , black!60, fill=black!40] (A-\i) at (0, 3-\i) [thick] {};
					}
					
					\foreach \j  in {1,...,2}{
						\node[draw,inner sep=0.1cm, circle , black!60, fill=blue!70] (C-\j) at (1.5, 3-\j) [thick] {};
					}
					
					\path[black!60] (A-1) edge[-] (C-1);
					\path[black!60] (A-2) edge[-] (C-2);
					\path[black!60] (A-2) edge[-] (C-1);

					\node[draw,inner sep=0.1cm, circle , black!60, fill=red!70] (S) at (-1.5, 1.5) [thick] {};	
					\node[draw,inner sep=0.1cm, circle , black!60, fill=black] (T) at (3, 1.5) [thick] {};
					\node at (-1.5, 1.9) {$s$};
					\node at (3, 1.9) {$t$};

					\path[black!60] (S) edge[-] (A-1);
					\path[black!60] (S) edge[-] (A-2);
					\path[black!60] (C-1) edge[-] (T);
					\path[black!60] (C-2) edge[-] (T);
					
					\node at (-0.7, 2) {$2$};
					\node at (-0.7, 1) {$4$};
					\node at (0.75, 2.2) {$\infty$};
					\node at (0.75, 0.8) {$\infty$};
					\node at (2.2, 2) {$2$};
					\node at (2.2, 1) {$4$};
					\node[shift={(0.60, 1.62)}, rotate=35] {$\infty$};
					\node[scale=1.2] at (-0.5, 2.8) {$\vec{u}=\left(-3,-2\right):$};
				\end{tikzpicture}
				\hskip 1cm
				\begin{tikzpicture}[transform shape,line width=2pt,scale=0.8]

					\foreach \i  in {1,...,2}{
						\node[draw,inner sep=0.1cm, circle , black!60, fill=black!40] (A-\i) at (0, 3-\i) [thick] {};
					}
					
					\foreach \j  in {1,...,2}{
						\node[draw,inner sep=0.1cm, circle , black!60, fill=blue!70] (C-\j) at (1.5, 3-\j) [thick] {};
					}
					
					\path[black!60] (A-1) edge[-] (C-1);
					\path[black!60] (A-2) edge[-] (C-2);
					\path[black!60] (A-1) edge[-] (C-2);

					\node[draw,inner sep=0.1cm, circle , black!60, fill=red!70] (S) at (-1.5, 1.5) [thick] {};	
					\node[draw,inner sep=0.1cm, circle , black!60, fill=black] (T) at (3, 1.5) [thick] {};
					\node at (-1.5, 1.9) {$s$};
					\node at (3, 1.9) {$t$};

					\path[black!60] (S) edge[-] (A-1);
					\path[black!60] (S) edge[-] (A-2);
					\path[black!60] (C-1) edge[-] (T);
					\path[black!60] (C-2) edge[-] (T);
					
					\node at (-0.7, 2) {$2$};
					\node at (-0.7, 1) {$4$};
					\node at (0.75, 2.2) {$\infty$};
					\node at (0.75, 0.8) {$\infty$};
					\node at (2.2, 2.1) {$2/3$};
					\node at (2.2, 0.8) {$16/3$};
					\node[scale=1.2] at (-0.5, 2.8) {$\vec{u}=\left(-3,-\frac{3}{2}\right):$};
					\node[shift={(0.60, 1.38)}, rotate=-35] {$\infty$};
					
				\end{tikzpicture}
			\end{center}
			\caption{The {networks and capacities for Example~\ref{ex_1} with budgets $\vec{b}=(-2,-4)$ and utility profiles computed in Example~\ref{ex_candidates}. Both networks admit a flow of magnitude $6$, so both utility profiles are competitive.}
				\label{fig_Networks_2_agent_unequal_budgets}
			}
		\end{figure}
	\end{example}
	
	\smallskip 
	
	\paragraph{\textbf{Algorithmic consequences.}} There are many efficient algorithms for solving maximum flow problems. For example, the Edmonds-Karp algorithm has linear runtime in the number of {nodes} and quadratic in the number of edges~\cite{Tardos_book}. {Lemma}~\ref{prop_MaxFlow} and Lemma~\ref{lm_price_net_utilities} thus yield the following algorithmic corollary. 
	
	\medskip  
	
	\begin{corollary}\label{cor_complexity_U_is_compet}
		Given a chore division problem $(\vec{v}, \vec{b})$ and a vector $\vec{u} \in \R_{<0}^n$, the following tasks can be completed in time $O((m+n)m^2n^2)$:
		\begin{itemize}
			\item checking whether $\vec{u}$ is a competitive utility profile, and
			\item if the answer is positive, computing a competitive allocation $\vec{z}$ with $\vec{u}(\vec{z})=\vec{u}$  and the corresponding vector of prices $\vec{p}$.
		\end{itemize}
	\end{corollary}
	\bigskip
	
	\begin{remark}[{Computing  \emph{all} competitive allocations}]\label{rem_computing_all}
		{The maximum flow algorithm reconstructs \emph{one} competitive allocation for a given competitive utility profile $\vec{u}$. Can we efficiently compute \emph{all} such allocations?
			The answer depend on whether the matrix $\vec{v}$ is degenerate (Definition~\ref{def_non_degenerate}) or not.}
		
		{For non-degenerate $\vec{v}$, for each Pareto optimal utility profile $\vec{u}$ there is a  \emph{unique} feasible allocation $\vec{z}$ such that $\vec{u}(\vec{z})=\vec{u}$ (Lemma~1 in \cite{BMSY_SCW}). Hence, the competitive allocation $\vec{z}$ recovered by the maximum flow algorithm is the only allocation with this utility profile. We conclude that for non-degenerate problems, the algorithm finds all competitive allocations.
			Note that if $n$ or $m$ are fixed, non-degeneracy of a matrix $\vec{v}$ can be checked in strongly polynomial time by inspecting each simple cycle (in a bipartite graph with parts $[n]$ and $[m]$ there are at most $(nm)^{\min\{n,m\}}$ of them\footnote{There are at most $n$ choices for the ``first'' {node} of a cycle, $m$ options for the second, $n-1$ options for the third plus one option to complete the cycle, etc. Therefore, there are not more than $(nm)^q$ cycles visiting each part of the graph at most $q$ times. For a simple cycle, $q$ cannot exceed the size of the smaller component. This leads to the upper bound $(nm)^{\min\{n,m\}}$ for the total number of simple cycles in a bipartite graph.}).}

		{For general $\vec{v}$, the set of feasible $\vec{z}$ with $\vec{u}(\vec{z})=\vec{u}$ is a convex polytope, which may contain more than one point if $\vec{v}$ is degenerate.  
			By computing a polytope $P$ we assume enumerating a finite number of points $p_1,\ldots,p_k$ such that $P$ is their convex hull.}
		
		{For degenerate $\vec{v}$, computing the set of all competitive allocations for a given competitive utility profile $\vec{u}$ becomes hard. The following example demonstrates that even for fixed $n$, the number of extreme points may be exponential in $m$, and hence enumeration of all of them requires at least an exponential number of operations.\footnote{We {claim that listing all the extreme points may take exponential time; however, this does not exclude the existence of a concise representation for this set. For example, such a representation is given by the phrase ``the set of extreme points of competitive allocations corresponding to a given utility profile $\vec{u}$''.}}}
		
		Two agents divide the set $[m]$ of identical chores: $v_{i,j}=-1$ for all $i,j$. Budgets are equal. The set of feasible utilities $\U$ is the linear segment between $(-m,0)$ and $(0,-m)$. The only critical point of the {Nash social welfare} $|u_1|\cdot|u_2|$ on this interval is the point $$\vec{u}=\left(-\frac{m}{2},-\frac{m}{2}\right)\,.$$ The set of all competitive allocations is thus the convex hull of all  $\vec{z}$ with $z_{i,j}\in\{0,1\}$  such that each row of the matrix $\vec{z}$ has $m/2$  non-zero elements, and each column has only one. 
		
		The number of such matrices is given by the binomial coefficient $C_{m}^{{m}/{2}}$, which grows exponentially for large $m$.
		\qed 
	\end{remark}

	\section{Finding approximately-fair allocations of indivisible chores}\label{sect_indiv}
	If items are indivisible, envy-free allocations may fail to exist (e.g., {if there are fewer items than} agents). This motivates considering relaxed fairness notions.
	
	Recently \cite{barman2018} described how to round a divisible competitive allocation with goods to get an approximately fair  Pareto optimal indivisible allocation. Their approach extends word by word to the case of chores. We will describe the main ingredients and refer to the original paper for the details of the construction.
	
	An allocation $\vec{z}$ is \emph{indivisible} if $z_{i,j}$ equals either $0$ or $1$. We will identify an indivisible bundle $\vec{z}_i$ with the set of those chores consumed by $i$, allowing us to use a set-theoretic notation. 
	\bigskip
	
	\begin{definition}[weighted-EF$_1^1$] 
		For a weight vector $\beta\in \R_{>0}^n$, an indivisible allocation $\vec{z}$ of chores is envy-free up
		to the removal of a chore from the first bundle and addition of another chore to the other bundle if for any pair of agents $i$ with non-empty bundle $\vec{z}_i$ and $i'$ there are two chores $j\in \vec{z}_i$ and $j'\in [m]$ such that 
		\begin{equation}\label{eq_EF11}
			\frac{u_i(\vec{z}_i\setminus\{j\})}{\beta_i}\geq  \frac{u_i(\vec{z}_{i'}\cup\{j'\})}{\beta_{i'}}\,.
		\end{equation}
	\end{definition}
	In the original definition for goods from \cite{barman2018}, an agent $i$ adds one item to their own bundle and throws away an item from the bundle of $i'$.
	
	Another popular fairness property, Proportionality, claims that every agent must prefer their bundle to equal division. Its relaxed version was introduced for the case of goods in~\cite{conitzer2017fair}. For chores, we have a mirror definition.
	\bigskip
	
	\begin{definition}[weighted-Prop$1$] 
		An indivisible allocation of chores $\vec{z}$ is weighted-proportional up to one chore  with a weight vector $\beta\in \R_{>0}^n$ if for each agent $i$ with non-empty bundle $\vec{z}_i$ there is a chore $j\in \vec{z}_i$ such that
		\begin{equation}\label{eq_prop1}
			{\frac{u_i(\vec{z}_i\setminus\{j\})}{\beta_i} \geq \frac{u_i([m])}{\sum_{i'=1}^n \beta_{i'}}\,.}
		\end{equation}
	\end{definition}
	
	The {following theorem shows that both fairness notions can be combined with Pareto optimality, and such an allocation can be efficiently computed.} 
	\bigskip
	
	\begin{theorem}\label{prop_indiv}
		{For any matrix of values with $\vec{v}\in \R_{<0}^{n\times m}$ and any vector of weights $\beta\in \R_{>0}^n$,} there exists an indivisible allocation $\vec{z}$ that is Pareto optimal in the divisible problem {(the utility profile $\vec{u}(\vec{z})$ belongs to the Pareto frontier $\U^*(\vec{v})$)} and satisfies {both} weighted-EF$_1^1$ and weighted-Prop$1$.
		If the number of agents $n$ or chores $m$ is fixed, then such an allocation can be found in strongly polynomial time.	
	\end{theorem}
The algorithmic part of the {theorem} is non-trivial only for the case of fixed $n$ (for fixed $m$, the total number of indivisible allocations $n^m$ is polynomial in $n$, and thus the exhaustive search gives a strongly polynomial algorithm).
	
	Theorem~\ref{prop_indiv} is a combination of Theorem~\ref{main} (the algorithm for computing a divisible competitive allocation) and the following theorem, which is the straightforward modification of a similar result of \cite{barman2018} for goods.
	\bigskip
	
	\begin{theorem}[Barman and Krishnamurthy \cite{barman2018}]\label{th_Barman}
		{Let $\vec{z}$ be a competitive allocation for values $\vec{v}\in \R_{<0}^{n\times m}$, budgets $\vec{b}\in \R_{<0}^n$, and prices $\vec{p}\in \R_{<0}^m$. There exists an indivisible competitive allocation $\vec{z'}$ with the same vector of prices $\vec{p}$ and a new vector of budgets $\vec{b}'\in \R_{-\leq 0}^n$ that is close to $\vec{b}$:} 
		\begin{itemize}
			\item For each agent $i$ with non-empty $\vec{z}_i'$, there are chores $j\in \vec{z}'_i$ and $j'\in [m]$ such that 
			\begin{equation}\label{eq_budgets_are_close}
				b_i-|p_j|\leq b_i'\leq b_i+|p_{j'}|.
			\end{equation}
			\item For agents $i$ with empty $\vec{z}_i'$, there is a chore  $j'\in [m]$ such that 
			\begin{equation}\label{eq_budgets_are_close_zero}
				b_i < b_i'=0 < b_i+|p_{j'}|.
			\end{equation}
		\end{itemize}
		The allocation $\vec{z}'$ can be computed in strongly polynomial time {given $\vec{z}$, $\vec{v}$, $\vec{b}$, and $\vec{p}$ as the input.} 
	\end{theorem}	
	\proof{Proof.}
	The proof repeats the one for goods (Theorem~3.1 in \cite{barman2018}) without substantial changes. We briefly describe the two main steps.
	
	The first step is to find a feasible allocation $\vec{z^\mathrm{acyc}}$ {with
		the same utility profile
		($\vec{u}(\vec{z})=\vec{u}(\vec{z}^\mathrm{acyc})$) and an acyclic consumption graph.} This can be done for any Pareto optimal allocation $\vec{z}$: {if the consumption graph of $\vec{z}$ has cycles, each of them} can be broken by a cyclic exchange that leaves every agent indifferent (see Lemma~1 in~\cite{BMSY_SCW} and the proof of {Lemma}~\ref{prop_crit_efficiency} in Appendix~\ref{app_Pareto_criteria}). An alternative construction for competitive allocations can be found in \cite{barman2018}.\footnote{Since every Pareto optimal allocation is competitive, the two approaches are equivalent.}
For almost all matrices $\vec{v}$, this cycle-elimination step is, in fact, redundant because the consumption graph of any Pareto optimal allocation $\vec{z}$ {has no cycles (see Corollary~\ref{cor_unit_cycles}).}
	By Lemmas~\ref{lm_MBB} and~\ref{lm_price_net_utilities}, the allocation $\vec{z}^\mathrm{acyc}$ is competitive with the same  $\vec{p}$ and $\vec{b}$.
	
	The second step is rounding the acyclic competitive allocation $\vec{z}^\mathrm{acyc}$ using  Algorithm~\ref{algo:round}.
	
	\medskip
	{
		\begin{algorithm}[H]
			\DontPrintSemicolon
			\KwIn{an acyclic consumption graph $G_{\vec{z}^\mathrm{acyc}}$, prices $\vec{p}\in \R_{<0}^m$, values $\vec{v}\in \R_{<0}^{n\times m}$, budgets $\vec{b}\in \R_{<0}^n$}
			\KwOut{an indivisible competitive allocation $\vec{z'}$}
			$\vec{z}'=0$;  $G=G_{\vec{z}^\mathrm{acyc}}$ \textcolor{cadetgrey}{\tcc{initialization}}
			\For {all chores $j$ connected to only one agent $i$ in $G$} {
				$\vec{z}_{i}'=\vec{z}_{i}'\cup \{j\}$\;
				chore $j$ is deleted from $G$\;
			}
			pick a root agent in each connected component
			and orient all edges from roots to leaves\;
			\While {$G$ has at least one agent $i$ with no incoming edges}{
				\While{there is a chore $j$ connected to $i$ such that $p_{j}+\sum_{j'\in \vec{z}_i'} p_{j'}  \geq b_i$}{
					$\vec{z}_{i}'=\vec{z}_{i}'\cup \{j\}$\;
					eliminate chore $j$ from $G$\;
				}
				\For {all remaining chores $j$ connected to $i$ in $G$} {				  $\vec{z}_{i'}'=\vec{z}_{i'}'\cup \{j\}$, where $i'$ is the agent that $j$ points to\;
					delete chore $j$ from $G$\;
				}
				agent $i$ is eliminated from $G$\;
			}
			\Return{$\vec{z}'$}\;
			\caption{Rounding of a competitive allocation}
			\label{algo:round}
		\end{algorithm}
	}
	\medskip
	{The allocation $\vec{z}'$ output by Algorithm~\ref{algo:round}} is competitive for prices $\vec{p}$ and budgets $b_i'=\sum_{j\in \vec{z}_i'}p_j.$ {Indeed, each agent $i$ consumes only those items that they consumed at the original allocation $\vec{z}^\mathrm{acyc}$} and thus the {MBB conditions from Lemma~\ref{lm_MBB}} are fulfilled. 
	
	By the construction, if $b_i<b_i'$ then there is always a chore $j$ linked to $i$ in $G_{\vec{z}^\mathrm{acyc}}$ but not allocated to $i$ at $\vec{z}_i'$ and such that $b_i'-|p_{j}|\leq b_i$; this gives the right-hand side of inequalities~\eqref{eq_budgets_are_close} and~\eqref{eq_budgets_are_close_zero}. Similarly, if $b_i> b_i'$, eliminating the last chore allocated to $i$ from their bundle changes the sign of the inequality, thus giving the left-hand side of~\eqref{eq_budgets_are_close}.  \qed
	\endproof
	\medskip
	Let us check that Theorems~\ref{th_Barman} and~\ref{main} together imply Theorem~\ref{prop_indiv}.
	\medskip
	
	\proof{Proof of Theorem~\ref{prop_indiv}.}
	By Theorem~\ref{main} we can compute a competitive  allocation $\vec{z}$ with the vector of budgets $\vec{b}=-\beta$ and  Theorem~\ref{th_Barman} shows that from $\vec{z}$ 
	we can build an {indivisible} competitive allocation $\vec{z'}$ with budgets satisfying~\eqref{eq_budgets_are_close} and~\eqref{eq_budgets_are_close_zero}. This allocation is Pareto optimal by the {First welfare theorem ({Theorem}~\ref{cor_welf}).} 
	
	For EF$_{1}^{1}$, consider an agent $i$ with non-empty $\vec{z}_i'$ and some agent $i'$. By~\eqref{eq_budgets_are_close}, we pick  $j\in \vec{z}_i'$ and $j'\in [m]$ such that $\sum_{c\in \vec{z}_i'\setminus\{j\}}p_c\geq -\beta_i$ and similarly $\sum_{c\in \vec{z}_{i'}'\cup\{j'\}}p_c\leq -\beta_j$. Therefore, a bundle $$\frac{\beta_i}{\beta_j}\Bigl(\vec{z}_{i'}'\cup\{j'\}\Bigr)\cup\{j\}$$ has a lower price than $\vec{z}_i'$. Thus inequality~\eqref{eq_EF11} follows from the fact that agent $i$ maximizes their utility among all bundles with a lower price. 
	
	To prove Prop$1$, we observe that the bundle $\frac{\beta_i}{\sum_{i'=1}^n \beta_{i'} }[m]$ has  price $-\beta_i$. By~\eqref{eq_budgets_are_close}, we can find $j\in \vec{z}_i'$ such that the price of $\vec{z}_i'\setminus\{j\}$ is higher than $-\beta_i$. Therefore, the bundle  $$\frac{\beta_i}{\sum_{i'=1}^n \beta_{i'} }[m]\cup\{j\}$$ is cheaper than bundle  $\vec{z}_i'$. Inequality \eqref{eq_prop1} follows since agent $i$ maximizes their utility on the budget constraint. \qed
	\endproof
	
	\smallskip 
	
	Next, we illustrate in an example with 2~agents and 3~chores  how to obtain the input required by Algorithm~\ref{algo:round} and then trace the execution of the algorithm.
	\smallskip

	\begin{example} \label{eg:example_for_alg3}
	{Let $[n] = \{1, 2\}$ and $[m] = \{1, 2, 3\}$, with budgets $\vec{b} = (-1, -2)$ and values $$\vec{v} =  \left(\begin{array}{ccc} -1 & -3 & - \frac{160}{17} \\  -2 & -5 & -15 \end{array}\right)\,.$$ 
			We first compute a competitive allocation and prices for the fractional problem, where the chores are divisible. 
			
			We observe that the chores are already sorted in weakly increasing order of the ratios $|v_{1,j}| / |v_{2,j}|$ for $j \in [m]$.
			We will compute a ``rich'' family of consumption graphs, each of which will be a  bipartite graph with parts $[n]$ and $[m]$. The edges of each graph will be obtained via  
		Lemma~\ref{prop_superset}, which constructs a rich family using $k/(k+1)$-splits and $k$-cuts as follows:}
		\medskip 
		\begin{enumerate}
			\item \emph{$k/(k+1)$-split, for $k=1,\ldots,m-1$ such that $|v_{1,k}|/|v_{2,k}| < |v_{1,{k+1}}| / |v_{2,{k+1}}|$: agent $1$ is linked to all chores $\{1,\ldots,k\}$ and agent $2$ is linked to all remaining $\{k+1,\ldots,m\}$. No other edges exist. }
			 We get the following graphs: \\ 
			\begin{enumerate}
				\item $k=1$: agent $1$ is linked to chore $1$; agent $2$ is linked to chores $2$ and $3$.
				\item $k=2$: agent $1$ is linked to chores $1$ and $2$; agent $2$ is linked to chore $3$. \\
			\end{enumerate}
			\item \emph{$k$-cut, for $k=1,\ldots,m$: agent $1$ is linked to chores $\{1,\ldots,k-1\}$, agent $2$ is linked to chores $\{k+1,\ldots,m\}$, while all chores $j$ with  $|v_{1,j}|/|v_{2,j}| = |v_{1,k}|/|v_{2,k}|$ are connected to both agents. No other edges exist. } Since, in our instance, the ratios $|v_{1,j}|/|v_{2,j}| $ are distinct for each $j$, we get the  graphs: \\ 
			\begin{enumerate}
				\item $k=1$: agent $2$ is linked to chores $2$ and $3$; both agents are linked to chore  $1$.
				\item $k=2$: agent $1$ is linked to chore $1$, agent $2$ is linked to chore $3$, and  both agents are linked to chore $2$. 
				\item $k=3$: agent $1$ is linked to chores $1$ and $2$; both agents are linked to chore $3$.
			\end{enumerate}
		\end{enumerate}
		\medskip 
		Suppose $\vec{z}$ is a competitive allocation. Let  $\vec{p}$ be the corresponding price vector. By Lemma~\ref{lm_price_net_utilities}, for each agent $i$ and chore $j$, the following constraints hold with respect to the tuple $(\vec{z}, \vec{p})$:
		\medskip 
$$				    
		    			p_j = \frac{v_{i,j} \cdot b_i}{u_i(\vec{z}_i)} \ \ \mbox{ if } \ \  z_{i,j}  > 0 \quad \mbox{ and } \quad p_j \geq \frac{ v_{i,j} \cdot b_i}{u_i(\vec{z}_i)} \ \  \mbox{ if } \ \ z_{i,j} = 0. 
		    			\eqno{(\bigstar)}
$$		\medskip 
		
		Checking condition $(\bigstar)$ in cases 1.a, 1.b, 2.a, 2.b, and 2.c above, we observe that case 2.c  yields a  bipartite graph compatible with a competitive allocation. This bipartite graph---denoted $G_{\vec{z}}^{\mathrm{acyc}}$---has parts $[n]$ and $[m]$, with agent-chore edges $\{(1,1), (1,2), (1,3), (2,3)\}$; see Figure~\ref{fig_G_acyc} for an illustration.
		
		\vspace{-1mm}
		\begin{figure}[h!]
			\centering
			
				\begin{tikzpicture}[transform shape,line width=2pt,scale=0.8]
					\foreach \i  in {1,2}{
						\node[draw,inner sep=0.1cm, circle , black!60, fill=black!40] (A-\i) at (0, 3.5-\i) [thick] {};
						\node at (-0.5, 3.5-\i) {\i};
						
					}
					
					\foreach \j  in {1,...,3}{
						\node[draw,inner sep=0.1cm, circle , black!60, fill=blue!70] (C-\j) at (1.5, 4-\j) [thick] {};
						\node at (2, 4-\j) {\j};
						
					}
					
					\path[black!60] (A-1) edge[-] (C-1);
					\path[black!60] (A-1) edge[-] (C-2);
					\path[black!60] (A-1) edge[-] (C-3);
					\path[black!60] (A-2) edge[-] (C-3);

					\node at (0, 3) {\emph{agents}};
					\node at (1.5, 3.5) {\emph{chores}};
					
				\end{tikzpicture}

			\caption{{The consumption graph $G_{\vec{z}}^{\mathrm{acyc}}$ for the values and budgets of Example~\ref{eg:example_for_alg3}. \label{fig_G_acyc}}}
		\end{figure}
		
		By solving the constraints in  $(\bigstar)$, we obtain  the  following  allocation, utilities, and prices: 

		\medskip 
		
		\begin{itemize}
			\item Allocation: $z_{1,1}  = 1, z_{1,2} = 1, z_{1,3} =  0.05$ and  $z_{2,1}  = 0, z_{2,2} = 0, z_{2,3} =  0.95$. \\ 
			\item Utilities: $u_1(\vec{z}_1)  = - {76}/{17} \; \; \; \mbox{and} \; \; \; u_2(\vec{z}_2) = -14.25$.\\
			\item Prices: $p_1  = \frac{v_{1,1} \cdot b_1}{u_1(\vec{z}_1)}  =  - \frac{17}{76}; \;\; \; p_2 = \frac{v_{1,2} \cdot b_1}{u_1(\vec{z}_1)} = - \frac{51}{76}; \;\; \; p_3 = \frac{v_{1,3} \cdot b_1}{u_1(\vec{z}_1)} = \frac{v_{2,3} \cdot b_2}{u_2(\vec{z}_2)} = - \frac{160}{76}$. \\
		\end{itemize}
		We now trace the execution of Algorithm 3 on this instance. The algorithm takes as input the values $\vec{v}$, the budgets $\vec{b}$, the consumption graph $G_{\vec{z}}^{\mathrm{acyc}}$, and the prices $\vec{p}$ above.
		
		 The first step of the algorithm is to initialize an allocation $\vec{z}'$ to zero (i.e., $\vec{z}_1' = \vec{z}_2' = \vec{0}$) and a graph $G$ to $ G_{\vec{z}}^{\mathrm{acyc}}$ (Line 1). 
		 
		 Then each chore $j$ connected to only one agent $i$ in the graph $G$ is assigned to agent $i$, and the chore $j$ is deleted from $G$. These are chores $1$ and $2$, which are only connected to agent $1$. Thus we set $\vec{z}'_{1,1} = 1$ and $\vec{z}'_{1,2} = 1$; then the vertices corresponding to chores $1$ and $2$ are deleted from the graph, together with   edges $(1,1)$ and $(1,2)$. The resulting graph $G$ now has  parts $\{1,2\}$ and $\{3\}$, with edges $\{(1,3), (2,3)\}$ (Lines 2-4);  part (a) of  Figure~\ref{fig:eg7.1_reduced_graph} shows an illustration.
		 \vspace{-2mm}
		 
		 	\begin{figure}[h!]
		 	
		 	\centering 

		 	\begin{tikzpicture}[transform shape,line width=2pt,scale=0.8]
					\foreach \i  in {1,2}{
						\node[draw,inner sep=0.1cm, circle , black!60, fill=black!40] (A-\i) at (0, 3.5-\i) [thick] {};
						\node at (-0.5, 3.5-\i) {\i};
						
					}

					\foreach \j  in {3,...,3}{
						\node[draw,inner sep=0.1cm, circle , black!60, fill=blue!70] (C-\j) at (1.5, 4-\j) [thick] {};
						\node at (2, 4-\j) {\j};
						
					}

					\path[black!60] (A-1) edge[-] (C-3);
					\path[black!60] (A-2) edge[-] (C-3);

					\node at (0, 3) {\emph{agents}};
					\node at (1.5, 2) {\emph{chores}};
					\node at (0.5, 0.5) {(a)};

				\end{tikzpicture}
		 	\hskip 1.5cm
		 	\begin{tikzpicture}[transform shape,line width=2pt,scale=0.8]
					\foreach \i  in {1,2}{
						\node[draw,inner sep=0.1cm, circle , black!60, fill=black!40] (A-\i) at (0, 3.5-\i) [thick] {};
						\node at (-0.5, 3.5-\i) {\i};
						
					}
					\node[draw,inner sep=0.1cm, circle , black!60, fill=red!80]  at (A-2) [thick] {};
					
					\foreach \j  in {3,...,3}{
						\node[draw,inner sep=0.1cm, circle , black!60, fill=blue!70] (C-\j) at (1.5, 4-\j) [thick] {};
						\node at (2, 4-\j) {\j};
						
					}

					\path[black!60] (A-1) edge[stealth-] (C-3);
					\path[black!60] (A-2) edge[-stealth] (C-3);

					\node at (0, 3) {\emph{agents}};
					\node at (1.5, 2) {\emph{chores}};
					\node at (0.5, 0.5) {(b)};

				\end{tikzpicture}
		 	\caption{{Figure (a) shows the graph $G$ obtained after deleting the chores connected to only one agent in  $G_{\vec{z}}^{\mathrm{acyc}}$. Figure (b) shows the same graph $G$ after the edges are oriented with agent $2$ designated as the root. \label{fig:eg7.1_reduced_graph}}}
		 \end{figure}
		 
		 \vspace{-2mm}
		 Next, Algorithm~\ref{algo:round} designates a root agent in each connected component of $G$ and orients all the edges from the roots to the leaves. Let agent $2$ be the root. Orienting the edges gives a directed graph with directed edges $2\to 3$ and $3 \to 1$. The resulting directed graph is illustrated in part (b) of Figure~\ref{fig:eg7.1_reduced_graph}.

		Then, while the graph $G$ has at least one agent $i$ with no incoming edges, and there is at least one chore $j$ connected with $i$ such that $p_j + \sum_{j' \in \vec{z}_{i}'} p_{j'} \geq b_i$, the chore $j$ is added to the allocation of agent $i$ (Lines 6-13). In our case, agent $i=2$ and chore $j=3$ meet these conditions  since 
		\smallskip 
		
			\begin{itemize}
				\item agent $2$ has no incoming edge,
				\item chore $3$ is connected to agent $2$, and 
				\item  $p_2 + \sum_{j' \in \vec{z}_{2}'} p_{j'} \geq b_2$ since $p_3 = - {160}/{76}$, $b_2 = -2$, and $\sum_{j' \in \vec{z}_{2}'} p_{j'} = 0$. 
			\end{itemize}
		\smallskip 
		  
		The algorithm adds chore $3$ to  the allocation of agent $2$, i.e., $\vec{z}_{2,3}' = 1$. Then chore $3$ is deleted from the graph (Line 9). No other chores are connected to agent $2$, and so agent $2$ is also deleted from the graph (Line 13).
		
		Finally, the while loop (Line 6) detects no other agent without incoming edges since agent $1$ has an incoming edge, and the algorithm finishes execution. Agent $1$ gets chores $1$ and $2$, while agent $2$ gets chore $3$.	
		\qed
	\end{example}
\color{black}
	
	\section{Concluding remarks}\label{sect_conclusions}
	
	Our main contribution is the new approach for computing competitive allocations for economies, where the set of competitive allocations may be disconnected. 
	
	\medskip 
	
	Several directions remain open. First, whether one { competitive allocation of chores} can be computed in polynomial time {if neither $n$ nor $m$ are fixed remains unknown.} Second, we have seen that computing all competitive utility profiles is not harder for economies with chores than computing the Pareto frontier. This could be an example of a general effect, and we expect that {our approach} extends to other division problems where the Pareto frontier can be computed in polynomial time.
	
	\medskip
	
	\textbf{Mixed Problems}. { The extension of our approach to the case of goods is straightforward and simpler, conceptually and in implementation, than existing algorithms.}
We expect that
		our approach also generalizes to mixed problems with goods and chores.
	
	\medskip
	
	\textbf{Fair Assignment and {constrained economies.}} In the fair assignment problem, the set of feasible allocations is restricted to those that satisfy an additional ``lottery'' constraint: for any agent { $i\in [n]$ we have $\sum_{j\in [m]} z_{i,j}={m}/{n}$ (for $n=m$ this means that each agent $i$ receives a lottery on $[m]$).} The competitive rule for fair assignment was studied by~\cite{HZ}, and almost forty years later, the first algorithm was proposed { in~\cite{Tardos} with the same performance as our algorithm for chores (polynomial in $n$ for fixed $m$ and in $m$ for fixed $n$).}
	{The algorithm  of~\cite{Tardos} is based on the black box of the cell enumeration technique, while our approach could plausibly give an alternative explicit construction.}

    {More generally, agents may have individual caps on the consumption of an item $z_{i,j}\leq C_{i,j}$ or have caps on total consumption $\sum_{j\in [m]} z_{i,j}\leq C_i$. Computing competitive allocations in these constrained economies is an open problem, which perhaps can be attacked using our approach.
	The main question is whether, in these settings, there is a polynomial algorithm enumerating faces of the Pareto frontier as in Section~\ref{sect_compute_Pareto}.}

	\medskip
	
	{ \textbf{Other applications.} The technique of computing the Pareto frontier seems to be a useful tool beyond applications to the competitive rule. 
		It was recently used by~\cite{FedorErel2019} to construct fair Pareto optimal allocations with a minimal number of shared items.  
		
		In general, the approach can be used for building efficient algorithms when a certain objective function (e.g., social welfare or the number of shared goods) is to be {optimized} over the set of Pareto optimal allocations under fairness constraints.
	}

\newpage 
	
\nocite{*}
\addcontentsline{toc}{section}{\protect\numberline{}References}%

\bibliographystyle{alpha}

\bibliography{main}

 \newpage
	\appendix
		\section{Characterization of competitive allocations}\label{app_sect_characterization}
		This auxiliary section starts from well-known formulas relating competitive price vectors and utilities (see \cite{AGT_book} chapter 5 and \cite{BMSY_SCW}). Using them, we derive three characterizations of competitive allocations that do not involve prices: by a system of inequalities, as maximizers of a linear objective, and as critical points of the {Nash social welfare}. 
		These results are not new (proved in \cite{BMSY_SCW} for equal budgets), but since we use them throughout the paper, we present short proofs here for the reader's convenience. 
		\subsection{Formulas for equilibrium prices}\label{app_subsect_MBB}
		At a competitive allocation $\vec{z}$, the {bundle} $\vec{z}_i$ of an agent $i$ maximizes their utility $u_i(\vec{z}_i)$ on a budget constraint $\sum_{j=1}^m p_j z_{i,j}\leq b_i$. Therefore, she consumes only those chores that have maximal utility to price ratio $v_{i,j}/|p_j|$ {(in the case of goods, this is known as the Maximal Bang per Buck property).} 
		
		The first-order conditions {for the utility-maximization on the budget constraint} imply the following characterization of competitive allocations.
		\begin{lemma}\label{lm_MBB}
			Fix $\vec{v}$ and $\vec{b}$ with strictly negative elements. A feasible allocation $\vec{z}$ is competitive for a vector of prices $\vec{p}\in \R_{<0}^n$ if and only if
			\medskip 
			\begin{itemize}
				\item {{MBB} condition:} $z_{i,j}>0 \Rightarrow \frac{v_{i,j}}{|p_j|}\geq \frac{v_{i,c}}{|p_c|}$ for all chores $c\in[m]$. \\
				\item {Budget exhaustion:} $b_i=\sum_{j=1}^m p_j z_{i,j}$.
			\end{itemize}
		\end{lemma}
		From this lemma, one can easily deduce formulas for prices in terms of $\vec{z}$ and $\vec{v}$.
		\begin{lemma}\label{lm_price_net_utilities}
			Consider a competitive allocation $\vec{z}$ for a division problem with a matrix of values $\vec{v}$ and a budget vector $\vec{b}$ having strictly negative elements.
			Then for any agent $i$ and chore $j$ we have:
			\medskip 
			\begin{itemize}
				\item $p_j=\frac{v_{i,j}b_i}{u_i(\vec{z}_i)}$ if $z_{i,j}> 0$.\\
				\item $p_j\geq \frac{v_{i,j}b_i}{u_i(\vec{z}_i)}$ if $z_{i,j}=0$.
			\end{itemize}
		\end{lemma}
		\proof{Proof.}
		By Lemma~\ref{lm_MBB}, agent $i$ spends their budget $b_i$ on chores $j$ with  maximal $v_{i,j}/|p_j|$. Therefore $i$'s utility  is  $u_i(\vec{z}_i)=|b_i|\max_{j\in [m]} v_{i,j}/|p_j|$. It follows that
		\medskip 
		
		\begin{itemize}
			\item ${u_i(\vec{z}_i)}/{|b_i|}= v_{i,j}/|p_j|$ if $z_{i,j}>0$ \\ 
			\item ${u_i(\vec{z}_i)}/{|b_i|}\leq v_{i,j}/|p_j|$, otherwise.
		\end{itemize} 
		\medskip 
		This completes the proof. \qed
		\endproof
		\medskip 
		
		As a corollary of Lemma~\ref{lm_price_net_utilities}, we see that {the utility profile $\vec{u}(\vec{z})$ of a competitive allocation $\vec{z}$} has strictly negative components and that $\vec{z}$ and $\vec{b}$ together uniquely determine $\vec{p}$.
		\subsection{Analogue of the Eisenberg-Gale result and other characterizations}\label{app_subsect_EisenbergGale}
		
		\begin{lemma}\label{prop_characterizations_of_CA}
			Fix a matrix $\vec{v}$ and a vector of budgets $\vec{b}$, both with strictly negative components, and consider an allocation $\vec{z}$. The following statements are equivalent:
			\medskip 
			
			\begin{enumerate}
				\item The allocation $\vec{z}$ is competitive.
				\item (\textbf{characterization by inequalities}) The utility profile $\vec{u}(\vec{z})$ has strictly negative components and $z_{i,j}>0$ implies
				\begin{equation}\label{eq_characterization_by_inequalities}
					\frac{v_{i,j}b_i}{u_i(\vec{z}_i)}\geq \frac{v_{{i'}j}b_{i'}}{u_{i'}(\vec{z}_{i'})}\quad \mbox{for all} \  \ i'\in[n].
				\end{equation}
				\item (\textbf{variational characterization}) The utility profile $\vec{u}(\vec{z})$ has strictly negative components and  
		{$$\vec{z}\in\argmax_{\scriptsize{\mbox{feasible }}\vec{y}} \sum_{i=1}^n  \frac{b_i}{u_i(\vec{z}_i)}\cdot u_i(\vec{y}_i),$$
				i.e., 
				$\vec{z}$ has the highest weighted utilitarian welfare $W_{\tau}$  among all feasible allocations, where the  weights $\tau$ are given by $\tau_i={b_i}/{u_i(\vec{z}_i)}$}  
				\item (\textbf{analogue of the Eisenberg-Gale characterization}) The vector $\vec{u}(\vec{z})$ has strictly negative components, belongs to the Pareto frontier $\U^*$, and is a critical point of the {Nash social welfare} {$\prod_{i=1}^n \left|u_i\right|^{|b_i|}$} on { the set of feasible utilities $\U$.}
			\end{enumerate}
		\end{lemma}
		\proof{Proof.}
		We will show that $(1)\Leftrightarrow(2)$, $(2)\Leftrightarrow(3)$, $(3)\Leftrightarrow(4)$.
		\begin{itemize}
			\item $(1)\Leftrightarrow(2)$: For competitive allocation $z$, inequalities~\eqref{eq_characterization_by_inequalities} hold because by Lemma~\ref{lm_price_net_utilities} we have $$\frac{v_{i,j}b_i}{u_i(\vec{z}_i)}=p_j\geq\frac{v_{{i'}j}b_{i'}}{u_{i'}(\vec{z}_{i'})}.$$
			
			In the opposite direction, if an allocation $z$ satisfies~\eqref{eq_characterization_by_inequalities} then we can define a ``candidate'' price vector by $$p_j=\frac{v_{i,j}b_i}{u_i(\vec{z}_i)}$$ for an agent $i$ that consumes a non-zero amount of $j$. Then~\eqref{eq_characterization_by_inequalities} implies {MBB} conditions of Lemma~\ref{lm_MBB}. The budget is exhausted since
			$$\sum_{j=1}^m p_j z_{i,j}=\sum_{j=1}^m \frac{v_{i,j}b_i}{u_i(\vec{z}_i)}z_{i,j}=\frac{b_i}{u_i(\vec{z}_i)}\sum_{j=1}^m v_{ij }z_{i,j}=b_i.$$ By Lemma~\ref{lm_MBB}, $\vec{z}$ is competitive.
			\item  $(2)\Leftrightarrow(3)$: An allocation $\vec{z}$ maximizes the weighted utilitarian welfare $W_\tau$ if and only if each chore $j$ is given to an agent with maximal weighted utility $\tau_i v_{i,j} $. Taking into account the definition of $\tau$, we see that these conditions are equivalent to the family of inequalities~\eqref{eq_characterization_by_inequalities}.
			
			\item  $(3)\Leftrightarrow(4)$:  The tangent hyperplane $h$ to the level curve of the {Nash social welfare} at $\vec{u}(\vec{z})$ is given by the equation $\vec{u}':\ \langle\vec{u}', \tau \rangle = \langle \vec{u}(\vec{z}), \tau \rangle$ (indeed,  $\tau$ is the gradient of  {$\ln\left(\prod_{i=1}^n \left|u_i\right|^{|b_i|}\right)$} and thus is orthogonal to the level curve). The hyperplane $h$ supports $\U$ iff $\vec{u}=\vec{u}(\vec{z})$ is either  minimum or maximum of $\langle \vec{u}, \tau \rangle $. The condition $\vec{u}(\vec{z})\in \U^*$ rules out the scenario with the minimum. Thus, $\vec{u}(\vec{z})$ is a critical point of the {Nash social welfare} on the Pareto frontier if and only if it maximizes $\langle\vec{u},\tau\rangle$ over the feasible set $\U$. Rewriting this condition in terms of the allocation $\vec{z}$, we get~$(3)$. \qed
		\end{itemize}
		\endproof
		
		\section{Criteria of Pareto optimality}\label{app_Pareto_criteria}
		Here we prove Lemma~\ref{prop_crit_efficiency}; the statement is repeated for the reader's convenience. 
		\medskip
		
		\begin{lemma*}[Lemma~\ref{prop_crit_efficiency}.]
			Let $\vec{v} \in \R^{n \times m}_{<0}$ be a matrix of values and $\vec{z}$ a feasible allocation. Then the following statements are equivalent:
			\begin{enumerate}
				\item The allocation $\vec{z}$ is Pareto optimal.
				\item The allocation {$\vec{z}$ has no} simple profitable trading cycles.
				\item There exists a vector of weights $\tau\in \R_{>0}^n$ such that the consumption graph {of $\vec{z}$} is a subgraph of {the MWW graph $G_\tau(\vec{v})$}.
				\item There exists a vector of weights $\tau\in \R_{>0}^n$ such that $\vec{z}$ maximizes the weighted utilitarian welfare {$\sum_{i\in [n]}\tau_i \cdot u_i(\vec{z}_i')$ 
					over all feasible allocations $\vec{z'}$.}
			\end{enumerate}
		\end{lemma*}
		\proof{Proof.}
		We will show the implications $(1)\Rightarrow(2)\Rightarrow(3)\Rightarrow(4)\Rightarrow(1)$.
		
		\begin{itemize}
			\item $(1)\Rightarrow(2):$ Let us show that if there is a simple profitable trading cycle $\C=(i_1,j_1, \ldots, i_{L+1}=i_1)$, then we can find a Pareto-improvement $\vec{z}'$ of the allocation $\vec{z}$. Indeed, transfer $\varepsilon_k$ amount of $j_k$ from $i_k$ to $i_{k+1}$ where the ratio of $\varepsilon_k$ and $\varepsilon_{k+1}$ comes from the condition that all agents $i_{k+1}$, $k=1, \ldots, {L-1}$ are indifferent between $\vec{z}$ and $\vec{z}'$: $$v_{i_{k+1}, j_{k}}\cdot \varepsilon_{k}=v_{i_{k+1}, j_{k+1}}\cdot \varepsilon_{k+1}\,.$$ These conditions define epsilons up to a multiplicative constant which can be selected small enough to guarantee feasibility. Profitability of $\C$ implies that agent $i_1$ is strictly better off: $$\varepsilon_L \cdot v_{i_{1},j_L}- \varepsilon_1 \cdot v_{i_{1},j_1}=\varepsilon_L \cdot v_{i_{1},j_L}(1-\pi(\C))<0,$$ and thus $\vec{z}'$ dominates $\vec{z}$.
			
			\item $(2)\Rightarrow(3):$ Let $i_1$ be an agent with a non-empty bundle $\vec{z}_{i_1}\ne \vec{0}$. Set its weight to $\tau_{i_1}=1$. For other agents $i$ define $\tau_i$ as $\max \; \pi(\P_{i_1, i})$ where the maximum is taken over all paths $\P_{i_1, i}=(i_1,j_1,\ldots,j_{L},i_{L+1}=i)$ connecting $i_1$ and $i$ such that $z_{i_k, j_k}>0$ for all $k=1,\ldots,L$. { The set of such paths is non-empty: for example, it contains a path $(i_1,j_1,i)$ for any chore $j_1$ with $z_{i_1,j_1}>0$.} By statement $(2)$, eliminating cycles in $\P_{i_1,i}$ can only increase $\pi$, and thus the maximum is finite and is attained on an acyclic path. 
			
			Consider the consumption graph $G_{\vec{z}}$ and let $i \in [n]$ and $j \in [m]$ be an agent and a chore that are 
			connected by an edge in $G_{\vec{z}}$. We show that they are also connected in $G_\tau$. For this we must check that $\tau_i \cdot |v_{i,j}|\leq \tau_{i'} \cdot |v_{i',j}|$ for each agent $i'$. Consider an optimal path $\P^*_{i_1, i}$ and extend it to a path $\P_{i_1, i'}$ by adding two extra {nodes} $j$ and $i'$. By definition of $\tau$ we get
			$$\tau_{i'}\geq \pi(\P_{i_1, i'})=\pi(\P^*_{i_1, i})\cdot\frac{|v_{i,j}|}{|v_{i',j}|}=\tau_i\cdot \frac{|v_{i,j}|}{|v_{i',j}|},$$
			which is equivalent to the desired inequality.
			
			\item $(3)\Rightarrow(4):$  Statement $(3)$ ensures that each chore $j$ is given to an agent $i$ with the lowest weighted disutility $\tau_i \cdot |v_{i,j}|$. Therefore, $\vec{z}$ has the maximal weighted welfare $W_\tau$ among all feasible allocations.
			
			\item $(4)\Rightarrow(1):$ If $\vec{z}'$ Pareto dominates $\vec{z}$, then $W_\tau(\vec{z}')>W_{\tau}(\vec{z})$. Thus, the maximizer of $W_\tau$ is undominated, which gives Pareto optimality. \qed
		\end{itemize}
		\endproof

\end{document}